\documentclass{IEEEtran}
\usepackage{amsmath,amssymb,amsthm}
\usepackage{bm}
\usepackage{graphicx}
\usepackage[utf8]{inputenc}
\usepackage{color}
\usepackage{comment}
\usepackage{bbm}
\usepackage{tikz}
\usetikzlibrary{arrows.meta, positioning, fit}
\usepackage{caption}
\usepackage{algorithm}
\usepackage{algorithmic}
\usepackage{dsfont}
\usepackage{hyperref}
\usepackage{todonotes}
\graphicspath{{fig}}
\usepackage{subcaption} 
\usepackage{paralist} 
\usepackage{ifthen}

\newif\ifdraftmode
\draftmodetrue

\ifdraftmode

    \newcommand{\pwc}[1]{\todo[color=yellow,size=\small,inline]{PW: #1}}
    \newcommand{\amc}[1]{\todo[color=orange,size=\small,inline]{AM: #1}}
    \newcommand{\dgc}[1]{\todo[color=purple!40,size=\small,inline]{DG: #1}}
\else

    \newcommand{\pwc}[1]{}
    \newcommand{\amc}[1]{}
    \newcommand{\dgc}[1]{}
\fi

\newcommand{\idx}{\iota}

\newtheorem{theorem}{Theorem}
\newtheorem{lemma}{Lemma}
\newtheorem{definition}{Definition}

\DeclareMathOperator*{\maximize}{maximize}
\DeclareMathOperator*{\minimize}{minimize}
\DeclareMathOperator*{\subjectto}{subject~to}

\begin{document}
\title{Bitcoin Mempool Linearization}
\author{Arman Mollakhani$^{1}$, Pieter Wuille$^{2}$, and Dongning Guo$^{1}$\\
$^{1}$Northwestern University, Evanston, IL \\
\{arman.mollakhani, dguo\}@northwestern.edu\\
$^{2}$Chaincode Labs, New York City, NY\\
pieter@wuille.net}
\date{}
\maketitle
\sloppy
\thispagestyle{plain}

\begin{abstract}
In the Bitcoin system, transactions arrive continuously at miners’ mempools and await inclusion in future blocks. Every non-coinbase transaction must spend one or more unspent outputs created by previous transactions, inducing dependency constraints among transactions in the mempool. At the same time, miners are economically incentivized to prioritize transactions with higher fee rates, measured as transaction fee per unit size. This paper formulates the \emph{mempool linearization problem}: given a set of transactions with associated fees, sizes, and dependency relationships, compute a dependency-respecting transaction ordering that maximizes fee-rate efficiency while supporting efficient updates as the mempool evolves dynamically. The problem is characterized through a partition of transactions into disjoint dependency-respecting subsets ordered by decreasing aggregate fee rate, together with an equivalent linear programming formulation. Motivated by structural properties of basic feasible solutions in the simplex method, a new algorithm called \emph{spanning forest linearization} (SFL) is developed. Operating directly on the transaction dependency graph, SFL iteratively merges and splits chunks of transactions to refine a global ordering, and is guaranteed to terminate at an optimal solution. Evaluation on both synthetic and real-world Bitcoin mempool data shows that SFL consistently computes optimal linearizations with substantially lower runtime than competing approaches, including a method based on the parametric preflow algorithm of Gallo, Grigoriadis, and Tarjan. These results indicate that SFL provides a practical and scalable framework for transaction prioritization by decentralized miners in large and rapidly evolving mempools. SFL has also been incorporated into the Bitcoin Core codebase for transaction cluster linearization.
\end{abstract}

\section{Introduction}
\label{s:introduction}

Efficient management of unconfirmed transactions is a fundamental challenge in blockchain systems such as Bitcoin. Newly created transactions are propagated through the peer-to-peer network and temporarily reside in each node's mempool while awaiting inclusion into future blocks. Because block space is limited and transaction fees provide economic incentives for prioritization, nodes must continuously determine which transactions are most profitable to relay, retain, evict, or include in blocks. As block construction and mempool management directly affect mining profitability, network performance, and security, efficient transaction ordering has become a fundamental algorithmic primitive in Bitcoin and related blockchain systems.

This task is complicated by transaction dependencies: every non-coinbase transaction must spend one or more unspent outputs created by previous transactions, and therefore cannot be included in a blockchain before all of its ancestors. In this paper, we formally define the \emph{mempool linearization problem}. We say a transaction in a node's mempool 
is valid if all of its ancestors are either in the same mempool or are already in the node's blockchain or ledger. Each transaction has an associated size (called \emph{weight} in 
Bitcoin)
and a fee paid to the miner of its host block. For any subset of transactions, the ratio between total fee and total size is referred to as its \emph{fee rate}.
The mempool linearization problem consists of partitioning the transactions in a mempool dependency graph into an ordered sequence of disjoint dependency-respecting subsets, or \emph{chunks}, arranged in descending order of aggregate fee rate. Dependency-respecting means that every transaction’s ancestors within the mempool must appear in earlier chunks.

An optimal linearization is important not only for high-revenue block construction, but more broadly for maintaining the efficiency and decentralization of the network itself. Nodes must continuously reason about transaction profitability throughout a transaction's lifecycle, including fee estimation before the transaction is fully constructed, relay prioritization during propagation, eviction decisions under %
mempool congestion, block template construction by miners, and mempool reconciliation after blockchain reorganizations. These decisions affect whether the public peer-to-peer network can reliably propagate transactions that miners are economically motivated to include.

This is directly tied to Bitcoin's long-term censorship resistance. If the public relay network cannot efficiently communicate %
profitable transactions, economic incentives favor the creation of private relay systems and privileged transaction distribution channels, as has occurred in other blockchain ecosystems~\cite{pahari2025exclusive}. Over time, reliance on private relay infrastructure can increase centralization by making it more difficult for new or smaller miners to compete without access to proprietary transaction flows. By contrast, a decentralized mining ecosystem requires that miners be able to obtain sufficiently accurate market information directly from the public peer-to-peer network. The mempool serves precisely this role: it is the distributed marketplace through which transaction demand, fee pressure, and profitability information are publicly disseminated. %
Efficient and accurate transaction prioritization is not merely a performance optimization problem, but also a core mechanism for preserving %
open miner participation and decentralized block production.

The present work focuses on the mempool linearization problem itself rather than the computationally intractable task of selecting an optimal subset of transactions under the strict block-size knapsack constraint. This distinction is primarily motivated by practical considerations. In realistic mempools, optimal chunks are typically much smaller than the overall block capacity, making the determination of the full chunk \emph{ordering} more important than fine-grained optimization for specific block boundaries. Moreover, profitability reasoning frequently occurs before mining time, when the eventual block boundary is unknown because future transaction arrivals cannot be predicted.

The quest for optimal mempool linearization can draw upon established paradigms in 
computer science %
and operations research. One 
approach %
reduces linearization %
to a sequence of maximum-ratio closure problems solvable %
using %
\emph{parametric preflow} techniques, such as the Gallo-Grigoriadis-Tarjan (GGT) framework~\cite{gallo1989fast}.
Another formulation casts the problem %
as an {integer linear program}, for which we demonstrate that an exact binary solution can be recovered from an %
optimum of linear programming (LP) with relaxations.
While these approaches provide useful theoretical foundations, practical demands of real-time mempool management impose strict runtime requirements because the mempool evolves continuously under highly dynamic and potentially adversarial conditions, necessitating solutions that are exceptionally fast. 

To address these challenges, this paper develops a new %
algorithmic framework called \emph{spanning forest linearization (SFL)}. 
Inspired by the structural properties of basic feasible solutions in the simplex method, SFL operates directly on the transaction dependency graph and iteratively refines a global ordering through local merge and split operations on chunks of transactions. The algorithm is guaranteed to terminate at an optimal linearization. In addition,
SFL naturally functions as %
an ``anytime'' algorithm, %
producing valid %
progressively improving %
linearizations through %
execution. This property is particularly important in adversarial settings, where computational budgets may be limited and attackers may attempt to construct pathological dependency structures.
Because SFL distributes computational effort across %
the transaction graph and can incorporate randomized refinement strategies, it avoids concentrating optimization effort on only a small portion of the mempool, thereby improving robustness against adversarial manipulation.
Extensive numerical evaluations on both synthetic and real-world Bitcoin mempool data demonstrate that SFL consistently computes optimal linearizations with substantially lower runtime than competing methods, including approaches based on the GGT parametric preflow algorithm. SFL has also been incorporated into the Bitcoin Core codebase%
\footnote{See \url{https://github.com/bitcoin/bitcoin/pull/32545} for the patch introducing it.} for transaction cluster linearization, demonstrating its practical suitability for deployment in production blockchain systems.

The remainder of this paper is organized as follows. Sec.~\ref{s:related-work} reviews related work.
Sec.~\ref{s:problem-definition} formally defines the mempool linearization problem.
Secs.~\ref{s:approach-lp} and~\ref{s:approach-spl} describe the LP and SFL approaches in detail, respectively. Sec.~\ref{s:compare-approaches} presents numerical results.
Sec.~\ref{s:conclusion} summarizes our findings and discusses %
future research.

\section{Related Work}
\label{s:related-work}

The problem of optimally ordering a set of transactions to maximize fee efficiency is mathematically equivalent to the classical single-machine non-preemptive scheduling problem with precedence constraints, denoted $(1 \mid \mathrm{prec} \mid \sum_j w_j C_j)$. The foundational version of this scheduling problem was first studied by Smith~\cite{Smith1956VariousOF}. In this equivalent formulation, transaction sizes correspond to job processing times, and transaction fees correspond to job weights (or penalties for delay). Finding the exact optimal sequence for this problem is fundamentally NP-hard%
~\cite{lawler1978sequencing}. To tackle this intractability, Sidney~\cite{sidney1975decomposition} introduced a foundational decomposition algorithm that partitions the jobs into an ordered sequence of subsets, ensuring that jobs in earlier subsets always precede those in later ones in any optimal schedule. While finding the optimal sub-ordering within each subset remains NP-hard, Sidney's decomposition provides a rigorous framework for simplifying the broader sequencing task. Potts~\cite{potts2009algorithm} proposed a branch and bound algorithm for this scheduling problem, deriving new lower bounds from a zero-one programming relaxation. Chudak and Hochbaum~\cite{chudak1999half} presented a simplified LP relaxation for this scheduling problem that can be solved via a minimum cut computation.

The classical formulation of maximum flow in networks was first developed by Ford and Fulkerson~\cite{Ford1962FlowsIN}, who introduced augmenting path algorithms and established the max-flow min-cut theorem. An important shift in maximum flow algorithms came with the introduction of the \emph{preflow} concept by Goldberg and Tarjan~\cite{GT88Preflow}, which allows temporary violations of flow conservation and guides excess flow using valid distance labels. This method laid the foundation for the development of efficient preflow-push algorithms. Cheriyan and Maheshwari~\cite{cheriyan1988preflow} provided a detailed analysis of the preflow-push paradigm, introducing optimizations based on vertex selection strategies that improve practical performance. Picard and Queyranne~\cite{Picard01111982} explored applications of minimum cuts in combinatorial optimization, including closure and separation problems, and demonstrated their relevance in dynamic settings. These network flow and cut concepts are fundamental to one of the approaches we explore. 

Gallo, Grigoriadis, and Tarjan~\cite{gallo1989fast} introduced a fast algorithm for solving \emph{parametric maximum flow problems}. 
Their technique solves a sequence of related minimum cut instances by incrementally updating the flow and label states, leveraging the structural monotonicity across parameter values. This strategy is particularly effective for problems where a parameter (e.g., a fee rate) evolves incrementally, as in our setting. Notably, the breakpoints identified by the GGT algorithm in our problem setting correspond to the Sidney decomposition~\cite{sidney1975decomposition} when the problem is viewed as a single-machine scheduling instance, and thus to what we call chunks in the mempool linearization.
A recent line of work~\cite{beinesmincut2025} introduces a simpler algorithm for monotone parametric minimum cut that identifies breakpoints in order without relying on dynamic trees 
or advanced flow computations. This method provides a conceptually cleaner alternative to classical parametric flow approaches and has been shown to perform well in practice for sparse graphs.

The maximum-ratio closure problem, central to the GGT-based approach, is a
specific case of \emph{fractional combinatorial optimization}, a well-studied class of problems involving the maximization of ratio-type objectives. Foundational ideas in this space were introduced by Isbell and Marlow~\cite{isbell1956attrition}, who studied ratio-based objectives in discrete optimization. Building on this direction, Pardalos and Phillips~\cite{pardalos1991fractional} developed a globally convergent algorithm for fractional programming problems, establishing theoretical bounds and exploring both linear and quadratic cases. The reduction of the maximum-ratio closure problem to a sequence of min-cut computations, as employed in the GGT approach, is inspired by these prior works.

Another relevant line of work addresses the maximum weight closure problem, which forms the basis for solving the more general maximum-ratio closure problem. A key algorithm in this area was introduced by Hochbaum~\cite{hochbaum2001new}, inspired by the seminal Lerchs-Grossmann (LG) algorithm originally developed for the open-pit mining problem~\cite{lerchs1965optimum}. Hochbaum provides a detailed analysis and presents powerful new variants of the LG algorithm, which solves the maximum closure problem directly, without explicitly computing or maintaining a feasible flow. Instead, the algorithm operates on a ``normalized tree'' structure, iteratively merging and re-normalizing branches that represent sets of nodes.

Beyond specialized flow and fractional programming techniques, LP %
offers a general and powerful framework for solving optimization problems. The seminal simplex algorithm%
~\cite{nocedal2006numerical}, though potentially exponential in its worst-case~\cite{klee1972simplex}, is highly efficient for a wide range of practical problems. 
For many combinatorial optimization problems, LP relaxations can be exact, meaning their optimal solutions are integral and directly solve the original discrete problem. The LP approach leverages this by formulating the mempool linearization problem as an LP and proving that an optimal binary solution can be recovered, thus allowing the use of standard LP solvers.

These diverse optimization paradigms---parametric network flows, fractional programming, and %
LP---provide the theoretical underpinnings for the algorithms we develop and analyze for the mempool linearization problem. Each offers distinct advantages and insights into structuring and solving this complex task.

\section{Problem Formulation} %
\label{s:problem-definition}

\subsection{Mempool Linearization Problem}
\label{ss:mempool-linearization-problem}

If transaction \(c\) spends an output of transaction \(p\), then transaction \(c\) is said to depend on transaction \(p\), or, equivalently, transaction $p$ is a parent of transaction $c$. Since transactions with missing parents cannot be included in a block, we assume no transaction misses any parents without loss of generality. A transaction may have multiple children and/or multiple parents. Since transactions cryptographically commit to their parents, we can assume that valid transaction dependencies cannot form cycles.

Consider a node's mempool consisting of $n$ distinct transactions, indexed by $1,\dots,n$. We disregard dependencies on other transactions already included in the node's current longest chain of blocks, as those dependencies are already respected. We further assume that if multiple transactions spend the same output, at most one of them resides in the mempool. Consequently, when assembling blocks using transactions in the mempool, it is both sufficient and necessary to respect only the dependencies within the mempool.

Let $V=\{1,\dots,n\}$. These dependencies are represented by a DAG \(%
(V, E)\).  %
Throughout this paper, we use the terms \emph{vertex} and \emph{transaction} interchangeably. There are two equivalent but oppositely oriented ways to draw edges in this DAG:
\begin{itemize}
    \item {\em Parent-to-child representation:} An edge \((i, j) \in E\) means transaction \(i\) is a parent of transaction \(j\).
    
    \item {\em Child-to-parent representation:} An edge \((i, j) \in E\) means transaction \(j\) is a parent of transaction \(i\).
\end{itemize}
Both conventions appear in the literature for historical reasons. In the remainder of this paper, we explicitly state which orientation is used in each algorithmic context.

\begin{definition}[mempool]
    A mempool consists of a collection of $n$ transactions whose dependencies are represented by a DAG $(V,E)$, where each transaction $i\in V=\{1,\dots,n\}$ is associated with a fee $f_i \geq 0$ 
    and a size $s_i>0$. We denote the mempool as the tuple $(V,E,f,s)$ where $f = (f_1, \dots, f_n)$ and $s = (s_1, \dots, s_n)$.
\end{definition}

The total fee and size of a subset of transactions \( U \subseteq V \) are denoted as $f(U)=\sum_{i \in U} f_i$ and $s(U)=\sum_{i \in U} s_i$, respectively. Further, if $U$ is non-empty, its {\em fee rate} is %
\begin{align} \label{eq:freerate}
    r(U) = \frac{f(U)}{s(U)}.
\end{align}
As a convention, we set $r(\emptyset)=0$. For two %
sets of transactions $A$ and $B$, we write $A \succeq B$ if %
$r(A) \ge r(B)$, and $A \preceq B$ if %
$r(A) \le r(B)$. The strict versions $A \succ B$ and $A \prec B$ are defined analogously. The notion of fee rate extends to graphs: Given a graph $G'=(V',E')$, we let $r(G')=r(V')$.

\begin{definition}[closure]
\label{def:C}
    A non-empty subset $U \subseteq V$ is called a \emph{closure} in graph $G=(V,E)$ if, for every $i \in U$, all parents of $i$ are also contained in $U$. 
\end{definition}

The fundamental problem facing a node is to determine the order in which transactions should be included in future blocks. We formalize this as follows.

\begin{definition}[valid ordering]
    A \emph{valid ordering} of a set of transactions $U \subseteq V$ with respect to DAG $(V,E)$
    is a bijection $\sigma: U \to \{1,\dots,|U|\}$ such that for every edge $(i,j)\in E$ with $i,j \in U$ (in the parent-to-child representation), the parent appears before the child: $\sigma(i) < \sigma(j)$. A valid ordering of %
    mempool $(V,E,f,s)$ is a valid ordering of $V$ with respect to $(V,E)$.
\end{definition}

Given a valid ordering $\sigma$ of mempool $(V,E,f,s)$, we define the \emph{completion size} of transaction $i$ as
\begin{align}
\label{eq:completion-size}
    C_i(\sigma) = \sum_{\substack{j \in V:\, \sigma(j) \le \sigma(i)}} s_j,
\end{align}
which is the cumulative size of all transactions occupying positions $1$ through $\sigma(i)$ in the ordering.

Consider assembling a block according to the order $\sigma$, stopping when including the next transaction would exceed the block size limit $W > 0$.
Transaction $i$ is included in the block if and only if $C_i(\sigma) \le W$, so the total fee income is
\begin{align}
    \Pi(\sigma, W) = \sum_{i \in V} f_i \, \mathds{1}_{\{C_i(\sigma) \le W\}}
    \label{eq:income-indicator}
\end{align}
where %
$\mathds{1}_{\{C\}}$ denotes the indicator\footnote{If $C$ holds, $\mathds{1}_{\{C\}}=1$; if $C$ does not hold, $\mathds{1}_{\{C\}}=0$.} of condition $C$. To evaluate an ordering independently of the actual block size, suppose $W$ is modeled as a uniform random variable on the interval $[0,\, s(V)]$. Let us define the \emph{weighted completion size} %
as
\begin{align}
\label{eq:wsct}
    \Gamma(\sigma) = \sum_{i \in V} f_i \, C_i(\sigma) .
\end{align}
The expected fee income can be expressed as
\begin{align}
    \mathbb{E}[\Pi(\sigma, W)]
    &= \sum_{i \in V} f_i \, \Pr\!\big[C_i(\sigma) \le W\big] \\
    &= \sum_{i \in V} f_i \left( 1 - \frac{C_i(\sigma)}{s(V)} \right) 
    \label{eq:expected-income} \\
    &= f(V) - \frac{1}{s(V)} \Gamma(\sigma) .
    \label{eq:expected-income-expanded}
\end{align}

Minimizing $\Gamma(\sigma)$, which is equivalent to maximizing the expected fee income, is precisely the single-machine non-preemptive scheduling problem with precedence constraints, %
where transaction sizes correspond to job processing times and transaction fees correspond to job weights. For general precedence constraints, this problem is NP-hard~\cite{lawler1978sequencing}.
Sidney~\cite{sidney1975decomposition} established that in any instance of this problem, 
there exists an optimal valid ordering and a sequence of sets $\emptyset=B_0 \subset B_1 \subset \dots \subset B_m=V$ %
such that:
\begin{compactitem}
    \item $B_1,\dots,B_m$ are closures in $G$;
    \item defining $U_k = B_k \setminus B_{k-1}$, the sets are ordered by %
    fee rates: $U_1 \succeq U_2 \succeq \dots \succeq U_m$;
    and
    \item %
    in the valid ordering, all transactions in $U_k$ appear before all transactions in $U_{k+1}$, for every $k$.
\end{compactitem}
This decomposition reduces the global ordering problem into two subproblems: (i) identifying the sets and their ordering, and (ii) ordering transactions optimally within each set. The second subproblem remains NP-hard in general, since each set with its internal dependencies is itself an instance of a single-machine non-preemptive scheduling problem. However, in the Bitcoin mempool setting, where the transaction dependency graph typically consists of many small independent clusters, the inter-group ordering captures the dominant structure of the problem. We therefore focus on formalizing and solving this decomposition, which we refer to as the \emph{mempool linearization problem}.

\begin{definition}[partition]
\label{def:partition}
    A \emph{partition} of set $V$ is a sequence of disjoint, non-empty subsets $(U_1, \dots, U_m)$ of $V$ such that $U_1 \cup \dots \cup U_m = V$.
\end{definition}

\begin{definition}[mempool linearization and optimality]
\label{def:L}
    A linearization of mempool $(V,E,f,s)$ is a partition of $V$ into $(U_1, \dots, U_m)$, %
    such that:
    \begin{compactitem}
        \item $U_1 \succeq U_2 \succeq \dots \succeq U_m$,
        and
        \item for every $k \in \{1, \dots, m\}$, $U_k$ is a closure in the subgraph of $G$ induced by %
        $V \setminus B_{k-1}$, where $B_{k-1}=U_1\cup\dots\cup U_{k-1}$.
    \end{compactitem}
    Let $\mathcal{C}_k$ denote the entire set of closures in the subgraph of $(V,E)$ induced by $V\setminus B_{k-1}$. The linearization is {\em optimal} if, for every $k$, %
    $U_k$ achieves the maximum possible fee rate among all closures in $\mathcal{C}_k$.
    An optimal linearization is minimal if, for every $k$, $U_k$ is a closure of minimum size $s(U_k)$ among those achieving the maximum fee rate in $\mathcal{C}_k$. We extend the notion of linearization to allow graphs (and trees in particular): We say graphs $\big((U_1,E_1)$, \dots, $(U_m,E_m)\big)$ form a linearization of $(V,E,f,s)$ if $(U_1,\dots,U_m)$ is a linearization.
\end{definition}

An optimal linearization can be computed by iteratively selecting a maximum-ratio closure from the remaining subgraph until no transactions remain. Specifically, at step $k$, we solve
\begin{align}
\label{eq:max-feerate}
    \maximize_{U \in \mathcal{C}_k %
    } \; r(U) .
\end{align}
The ordered sequence of closures produced by this greedy procedure coincides with the Sidney decomposition~\cite{sidney1975decomposition} of the corresponding single-machine scheduling instance.

A minimal optimal linearization is obtained by selecting in~\eqref{eq:max-feerate}, at each step $k$, a closure of smallest size $s(U)$ among those attaining the maximum fee rate.

\begin{definition}[valid ordering induced by a linearization]
\label{def:induced-ordering}
    Let $(U_1, \dots, U_m)$ be a linearization of mempool $(V,E,f,s)$, and let $\sigma_k$ be a valid ordering of transactions in $U_k$ with respect to the subgraph induced by $U_k$ for each $k \in \{1,\dots,m\}$. The \emph{induced valid ordering} %
    is defined by
    \begin{align}
        \sigma(i) = \sigma_k(i) + \sum_{j=1}^{k-1} |U_j|, \quad \forall %
        i \in U_k.
    \end{align}
    This extends to graph linearizations via Definition~\ref{def:L}.
\end{definition}

A valid ordering $\sigma$ can also provide a trivial linearization consisting of singletons: $(\{\sigma(1)\},\dots,\{\sigma(n)\})$.

Throughout this paper, we use the shorthand $U^m = (U_1, \dots, U_m)$ for a sequence of sets. Accordingly, $f(U^k) = \sum_{i=1}^k f(U_i)$ and $s(U^k) = \sum_{i=1}^k s(U_i)$ for $k=1, \dots, m$, with the convention that $f(U^0) = s(U^0) = 0$.

To provide a quantitative metric for comparing linearizations, we introduce the geometric representation of the partitioned transaction sequence.

\begin{definition}[cumulative fee-size diagram]
    Given a partition of transactions (which may or may not be a linearization), denoted as $U^m = (U_1, \dots, U_m)$, let $C_{U^m}$ be the continuous piecewise-linear function obtained by linearly connecting the sequence of points 
    $(s(U^k), f(U^k) )$, $k=0,1,\dots,m$. The \emph{cumulative fee-size diagram} of $U^m$ is the diagram formed by the graph of function $C_{U^m}$ over its domain $[0, s(U^m) %
    ]$.
\end{definition}

\begin{definition}[area under the curve (AUC)]
    Given a partition of transactions, denoted as $U^m = (U_1, \dots, U_m)$, the area under the curve is defined as the definite integral of its cumulative fee-size diagram. It can be computed as 
    \begin{align}
        A(U^m) 
        &= \int_0^{s(U^m)} C_{U^m}(x) \, dx \\
        &= \
        \sum_{k=1}^m s(U_k) \left( f(U^{k-1}) + \frac{1}{2} f(U_k) \right).      \label{eq:AUmsum}
    \end{align}
\end{definition}

To establish the relationship between the AUC and optimal linearizations, we first establish several properties.

\begin{lemma} %
\label{lm:closures_properties}
    The non-empty intersection and union of two closures %
    are also closures. %
\end{lemma}
\begin{proof}
    Let \( U_1 \) and \( U_2 \) be two closures. %

    To prove that %
    $U_1 \cap U_2 \neq \emptyset$ is a closure, consider an arbitrary node %
    \( i \) in %
    both \( U_1 \) and \( U_2 \) and a child of $i$, denoted as $j$. Since \( U_1 \) is a closure, \( j \in U_1 \). Since \( U_2 \) is a closure, \( j \in U_2 \). Thus \( j \in U_1\cap U_2 %
    \).

    To prove that %
    $U_1 \cup U_2$ is a closure, consider an arbitrary node $i$ and its child $j$. %
    If \( i \in U_1 \), then \( j \in U_1 \). If \( i \in U_2 \), then \( j \in U_2 \). In either case, \( j \in U_1 \cup U_2%
    \).
\end{proof}

\begin{lemma}[prefix closure property]
\label{lem:prefix-closure}
    Let $%
    (U_1, \dots, U_m)$ be a linearization of $(V,E,f,s)$. For every $k \in \{1, \dots, m\}$, the set $U_1\cup\dots\cup U_k$ is a closure in $(V,E)$.
\end{lemma}
\begin{proof}
    Let $B_k=U_1\cup\dots\cup U_k$ for every $k$. Let $v$ be an arbitrary vertex in $B_k$. Evidently, $v\in U_j$ for some $j\in\{1,\dots,k\}$. Since
    $U_j$ is a closure in the subgraph induced by $V \setminus B_{j-1}$, %
    every parent of $v$ must either be in $U_j$ or in $B_{j-1}$. %
    Hence all $v$'s parents are in $U_j\cup B_{j-1}=B_j\subseteq B_k$, which implies that $B_k$ is a closure in $(V,E)$.
\end{proof}

\begin{lemma}[closure difference property]
\label{lem:closure-difference}
    Let $U$ and $W$ be two closures in $%
    (V,E)$ such that $W \setminus U \neq \emptyset$. Then $W \setminus U$ is a closure in the subgraph induced by $V \setminus U$.
\end{lemma}

\begin{proof}
    Let $v$ be an arbitrary vertex in $W \setminus U$ and $u$ be an arbitrary parent of $v$ in $V \setminus U$. Since $W$ is a closure, %
    $u$ must be in $W$. Therefore, $u \in W \cap (V \setminus U) = W \setminus U$.
\end{proof}

\begin{theorem}[AUC as an optimality metric]
\label{thm:auc-optimality}
    Let $L^* = (U^*_1, \dots, U^*_p)$ be an optimal linearization of %
    mempool $(V,E,f,s)$.
    Let $L = (U_1, \dots, U_q)$ be any other valid linearization of the same mempool. %
    Then the cumulative fee-size diagram of $L^*$ dominates that of $L$ everywhere; that is, $C_{L^*}(x) \ge C_L(x)$ for all $x \in [0, s(V)]$, which implies the AUC dominance:
    $A(L^*) \ge A(L)$. Consequently, $L^*$ maximizes the AUC among all linearizations of the mempool.
\end{theorem}

Theorem~\ref{thm:auc-optimality} is proved in Appendix~\ref{a:auc-optimality}.

\subsection{Desirable Algorithmic Properties}
\label{ss:desirable-properties}

An effective algorithm for mempool linearization should exhibit several key properties to %
perform well in a real-world, dynamic, and potentially adversarial environment.

First, it should function as an \emph{anytime algorithm}, %
producing a valid and progressively improving %
linearization throughout %
its execution. %
In practice, the algorithm is typically subject to %
strict time limits and may be %
terminated before reaching %
the global optimum. In such settings, the rate of quality improvement per unit of time is more important than the total time required to find the absolute optimum.

Second, the algorithm should exhibit %
granularity. If multiple closures achieve the same maximum fee rate, a \emph{smallest} closure is more desirable.
Smaller closures reduce the approximation %
gap between simple greedy %
ordering %
and the exact %
NP-hard knapsack-like block-packing problem, %
enabling more precise utilization of each block's capacity.
Hence we shall seek a \emph{minimal} optimal linearization.

Third, the solution must demonstrate robustness against mempool manipulations. Specifically, %
an attacker could %
construct complex transaction dependency graphs designed to exploit specific algorithmic weaknesses. 
For example, by attaching %
adversarial transactions with specific fees, sizes, and dependencies to a cluster of honest transactions,
the algorithm may be tricked into spending %
all available time attempting to improve the ordering of some selected %
transactions, leaving simple improvements to many other honest transactions %
unperformed when the time budget expires. To mitigate this, the algorithm should incorporate randomization and distribute computational effort fairly across all parts of the transaction graph, so that when execution is terminated early, all regions of the graph have received some optimization attention.

\subsection{Overview of %
Solutions}
\label{s:solutions-overview}

We investigate three distinct algorithmic paradigms for solving the mempool linearization problem, each offering a unique perspective on balancing optimality, performance, and practical utility.

We first describe an equivalent LP formulation for the maximum-ratio closure problem~\eqref{eq:max-feerate},
so that the full suite of LP solvers, including the simplex method and interior-point methods, can be applied to solve the mempool linearization problem exactly.

Our primary algorithmic innovation in this paper is the development of SFL, which is inspired by the structural properties of basic feasible solutions in the simplex method. SFL operates directly on the transaction graph, iteratively merging and splitting groups of transactions to refine a global linearization. SFL is explicitly designed to satisfy all the desirable algorithmic properties outlined in Sec.~\ref{ss:desirable-properties}.
Our empirical results show that SFL is exceptionally fast in practice, making it a strong candidate for real-time deployment.

To provide a theoretical foundation and benchmark for comparison, we also consider solving the maximum-ratio closure problem using the GGT parametric maximum flow framework. Since GGT's parametric breakpoint algorithm can be straightforwardly applied to our problem setting, we describe it in Appendix~\ref{s:approach-ggt} and focus on benchmarking its performance against SFL.

These approaches are compared empirically in Sec.~\ref{s:compare-approaches}, where we benchmark their performance on both synthetic and real-world mempool data.

\section{Linearization via LP} %
\label{s:approach-lp}

In this section, we present a solution to mempool linearization based on LP. We first formulate the search for a maximum-ratio closure~\ref{eq:max-feerate} in each step as a binary program and then show that it can be reduced to an equivalent linear program. Throughout this section, we consider the graph $G=(V,E)$ with the parent-to-child representation.

\subsection{Reduction to LP}
\label{ss:reduction-to-lp}

Without loss of generality, we consider the search for the first closure \( U_1 \subseteq V \) that maximizes the fee rate. Let $z_i\in\{0,1\}$ indicate whether transaction $i$ is included in $U_1$. Problem~\ref{eq:max-feerate} can be reformulated as a binary program (BP):
\begin{subequations}\label{eq:bp}
\begin{align}
r_{\text{BP}}^{*} = \maximize_{z_1, \dots, z_n} \quad & \frac{\sum_{i \in V} f_i z_i}{\sum_{i \in V} s_i z_i} \label{eq:bp:a} \\
\subjectto \quad & z_i \in \{0, 1\}, \quad \forall i \in V \label{eq:bp:b} \\
& \sum_{i \in V} z_i > 0 \label{eq:bp:c} \\
& z_i \ge z_j, \quad \forall (i,j)\in E . \label{eq:bp:d}
\end{align}
\end{subequations}
We require $z_i\ge z_j$ for every parent-child pair $(i,j)$, ensuring that no transaction is included without its parents.

Consider the following fractional program (FP):
\begin{subequations}\label{eq:fp}
\begin{align}
r_{\text{FP}}^{*} = \maximize_{y_1, \dots, y_n} \quad & \frac{\sum_{i \in V} f_i y_i}{\sum_{i \in V} s_i y_i} \label{eq:fp:a} \\
\text{subject to} \quad & y_i \ge 0, \quad \forall i \in V \label{eq:fp:b} \\
& \sum_{i \in V} y_i > 0 \label{eq:fp:c} \\
& y_i \ge y_j, \quad \forall (i,j)\in E . \label{eq:fp:d}
\end{align}
\end{subequations}
This FP can be viewed as the BP in~\eqref{eq:bp} with the binary constraints relaxed. In Appendix~\ref{a:lp=bp}, we show that the relaxation yields no increase in the maximum fee rate. Furthermore, we transform the FP to an equivalent linear program and establish the following result:

\begin{theorem} \label{thm:lp=bp}
    Given any strictly positive real numbers 
    $s_{1},...,s_{n}$
    and nonnegative numbers
    $f_{1},...,f_{n}$, %
    if $(x_{1},...,x_{n})$ solves the following LP,
\begin{subequations}\label{eq:lp}
\begin{align}
r^*_{\text{LP}} = \maximize_{x_{1},...,x_{n}} \quad & \sum_{i \in V} f_i x_i \label{eq:lp:a} \\
\subjectto \quad & \sum_{i \in V} s_i x_i = 1 \label{eq:lp:b} \\
& x_i \ge 0, \quad \forall i \in V \label{eq:lp:c} \\
& x_i \ge x_j, \quad \forall (i,j)\in E , \label{eq:lp:d}
\end{align}
\end{subequations}
then $(z_1, \dots, z_n)$ defined by %
\begin{align}
\label{eq:zi-def}
    z_i = \mathds{1}_{ \left\{ x_i\ge x_j, \forall j \in V \right\}} %
\end{align}
solves BP~\eqref{eq:bp}.
\end{theorem}

Theorem~\ref{thm:lp=bp} is proved in Appendix~\ref{a:lp=bp}.

In~\ref{eq:lp},
each transaction \( i \in V \) is associated with a non-negative real-valued variable \( x_i \ge 0 \) (called transaction variable), where 
\( x_i = \max_{j\in V}(x_j) \)
indicates that transaction \( i \) is included in the solution set, and 
excluded otherwise.

\begin{figure}
    \centering
    \begin{tikzpicture}[
        node distance=1.2cm,
        tx/.style={circle, draw, thick, minimum size=1.1cm, inner sep=2pt, align=center, font=\small},
        edge/.style={->, >=Stealth, thick},
        group1/.style={draw, dotted, thick, rounded corners=0.4cm, inner sep=0.2cm},
        group2/.style={draw, dotted, thick, rounded corners=0.5cm, inner sep=0.4cm}
    ]
        \node[tx] (1) {tx 1 \\ $0/1$};
        \node[tx, right=of 1] (2) {tx 2 \\ $20/1$};
        \node[tx, right=of 2] (3) {tx 3 \\ $0/1$};
        \node[tx, right=of 3] (4) {tx 4 \\ $20/1$};

        \node[group1, fit=(1) (2)] {};
        \node[group2, fit=(1) (2) (3) (4)] {};

        \draw[edge] (1) -- (2);
        \draw[edge] (2) -- (3);
        \draw[edge] (3) -- (4);
    \end{tikzpicture}
    \caption{An example of a transaction dependency graph. Each transaction $i$ is shown as a circle annotated with its fee and size: ``$f_i/s_i$''. The dotted boundaries %
    enclose the minimal optimal closure %
    and a larger optimal closure; %
    both have the same fee rate.}
    \label{fig:minimal-closure-example}
\end{figure}
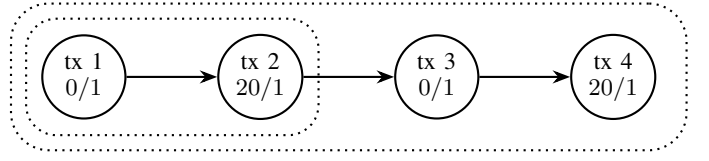

The LP formulation does not inherently guarantee a minimal optimal closure. To see this, consider an example transaction graph, illustrated in Fig.~\ref{fig:minimal-closure-example}, consisting of a set of four transactions, $V = \{1, 2, 3, 4\}$, with parent-to-child dependencies $E=\{(1,2),(2,3),(3,4)\}$. Let all sizes be $s_i = 1$, and the fees be $f_1 = 0$, $f_2 = 20$, $f_3 = 0$, and $f_4 = 20$. 

The maximum optimal fee rate is $r^* = 10$. This is achieved by two valid closures: the minimal closure $\{1, 2\}$, and the larger closure $\{1, 2, 3, 4\}$. 
The optimal value is achieved along a flat face of the feasible polytope. The vertices defining this optimal face are:
\begin{compactitem}
    \item vertex 1: $x_1 = %
    x_2 = 1/2$, $x_3 = %
    x_4 = 0$, and
    \item vertex 2: $x_1 = %
    x_2 = %
    x_3 = %
    x_4 = 1/4$.
\end{compactitem}

Any convex combination of these vertices, such as the edge midpoint ($x_1 = %
x_2 = 3/8$, $x_3 = %
x_4 = 1/8$), is also an optimal solution. The specific $x_i$ values returned depend entirely on the solver mechanism. For instance, an interior-point method might terminate on the midpoint; applying the exact extraction rule~\eqref{eq:zi-def} successfully isolates the minimal closure $\{1, 2\}$. 
However, the Simplex method exclusively explores the vertices of the polytope and will not return an intermediate edge value. It %
has no inherent preference to
terminate on either vertex 1 or vertex 2. If it terminates on vertex 1, the extraction rule isolates $\{1, 2\}$ (minimal). If it terminates on vertex 2, the maximum value is $1/4$ for all variables, meaning the extraction rule isolates the entire set $\{1, 2, 3, 4\}$ (not minimal). 

One way to isolate a minimal optimal closure is via a two-stage approach. Once the maximum optimal fee rate \( r^*_{\text{LP}} \) is found from the primary LP, we can formulate a secondary BP to explicitly minimize the total size of the closure, subject to the dependency constraints and the requirement that it achieves the optimal rate:

\begin{subequations}
\begin{align}
    \minimize_{z_1, ..., z_n} \quad & \sum_{i \in V} s_i z_i \\
    \subjectto \quad & \sum_{i \in V} f_i z_i \ge r^*_{\text{LP}} \sum_{i \in V} s_i z_i \\
    & z_i \ge z_j, \quad \forall (i,j)\in E \\
    & \sum_{i \in V} z_i \ge 1 \label{eq:non_empty} \\
    & z_i \in \{0, 1\}, \quad \forall i \in V .
\end{align}
\end{subequations}
The constraint~\eqref{eq:non_empty} ensures the solution is non-empty. This formulation strictly returns the indicator variables \( z_i \) for a smallest valid closure that matches the maximum fee rate.

In %
Sec.~\ref{ss:sfl-granularity}, we %
introduce a technique involving perturbation to guarantee the retrieval of a minimal optimal closure.

\subsection{Simplex Algorithm %
for Mempool Linearization}
\label{ss:simplex-alg}

The simplex algorithm for solving LPs %
operates on the vertices of the feasible polytope~\cite{nocedal2006numerical}.
Despite its exponential worst-case time complexity, the simplex algorithm is remarkably efficient in practice, often exhibiting polynomial-time behavior on real-world problems~\cite{spielman2004smoothed}.
The deterministic simplex algorithm's exponential complexity arises from pathological inputs%
~\cite{klee1972simplex}, where a fixed pivot rule is forced to traverse a long path of vertices. Randomized variants of the simplex algorithm circumvent this issue by modifying the pivot rule to incorporate randomness, for example, by choosing a variable uniformly at random from all eligible candidates that improve the objective function~\cite{kalai1992subexponential,matouvsek1992subexponential}.

To apply the simplex method, we first convert the inequality constraints~\eqref{eq:lp:d} into the standard form. %
Specifically, for each %
dependency \( (p%
, c%
) \in E \), a slack variable $d_{p,c}
$ is introduced to replace the constraint as
\begin{align}
    d_{p,c} &= x_p - x_c \label{eq:d=x-x} \\ %
    d_{p,c} &\ge 0.
\end{align}
Let $m=|E|$ denote the number of dependencies.
Along with the normalization constraint~\eqref{eq:lp:b},
we obtain a system with \( n + m 
\) variables and \( m + 
1
\) equations.
At each step, the simplex algorithm maintains a \textit{basic feasible solution} (BFS), corresponding to a vertex of the feasible polytope, by partitioning the variables into two disjoint sets:
\begin{itemize}
    \item \emph{$n-1$ non-basic (free) variables}: These are %
    not currently in the basis and %
    set to zero. %
    They are considered ``free'' because they are eligible to enter the basis during pivot steps. %
    \item \emph{$m+
    1
    $ basic variables}: %
    These are part of the current basis (solution), %
    which are uniquely determined by %
    the system of \( m + 
    1\) equations %
    with the free variables set to zero. %
\end{itemize}

At each iteration, the algorithm performs a \textit{pivot operation} by selecting one non-basic variable to \textit{enter the basis} (become basic and potentially positive) and one basic variable to \textit{leave the basis} (become non-basic and set to zero). This transition moves from one vertex of the feasible polytope to an adjacent vertex, ideally improving the objective function value.

\begin{lemma}[uniformity] %
\label{lem:bfs-uniformity}
    At any BFS of LP~\eqref{eq:lp}, all strictly positive 
    variables take the exact same value.
\end{lemma}

\begin{proof}
At any BFS, exactly $n+m$ linearly independent constraints must be active in the LP's standard form with $n$ transaction variables and $m$ slack variables. Since the $m+1$ equality constraints are always active, exactly $(n+m) - (m+1) = n-1$ of the non-negativity constraints (on both transaction and slack variables) must be active~\cite{bertsimas1997introduction}.

Let $Z \subseteq V$ be the index set of active transaction constraints ($x_i = 0$), and let $\mathcal{D} \subseteq E$ be the set of edges which correspond to active slack constraints ($d_{p,c} = 0$), such that $|Z| + |\mathcal{D}| = n - 1$. For every $(p,c)\in\mathcal{D}$, substituting $d_{p,c} = 0$ into \eqref{eq:d=x-x} yields an equation $x_p - x_c = 0$.
For the basis of the BFS to be linearly independent, the edge set $\mathcal{D}$ cannot contain cycles.
If a cycle existed linking vertices $v_1, \dots, v_k$, the sum of their corresponding active slack equations $(x_{v_1} - x_{v_2}) + \dots + (x_{v_k} - x_{v_1}) = 0$ would demonstrate linear dependence. Thus, $\mathcal{D}$ forms a forest on the vertices of $V$. The
forest with $n$ vertices and $|\mathcal{D}|$ edges contains exactly $n - |\mathcal{D}|$ trees~\cite{diestel2017graph}. Substituting $|\mathcal{D}| = n - 1 - |Z|$ yields exactly $|Z| + 1$ trees. Within each tree, the active constraints $x_p - x_c = 0$ force all transaction variables $x_i$ to a single uniform value.

To maintain linear independence of the basis, no two of the $|Z|$ active constraints $x_i = 0$ can fall in the same tree. Thus, exactly $|Z|$ distinct trees are forced to have $x_i = 0$, leaving exactly one tree, denoted $C^*$, to have a strictly positive uniform transaction value $v > 0$.
Moreover, %
for any dependency $(p,c) \in E$, the slack variable $d_{p,c} = x_p - x_c$ is strictly positive if and only if $p \in C^*$ and $c \notin C^*$, which yields $d_{p,c} = v - 0 = v$. Therefore, every strictly positive variable in the BFS evaluates to the exact same value $v$.
\end{proof}

Given a BFS, we say two transactions $p$ and $c$ are in the same ``chunk'' if $d_{p,c}$ is a free slack variable, i.e., $d_{p,c}=0$ is an active dependency constraint. By extension, a chunk includes a set of transactions that are connected through free slack variables. 
Due to additional structural constraints imposed by the simplex algorithm, several properties follow.

By Lemma~\ref{lem:bfs-uniformity}, a BFS determines the roles of each variable. A transaction variable \(x_i > 0\) is included in the current best closure, while \(x_i = 0\) is excluded from it. 
A slack variable with \( d_{p,c} = 0 \) enforces that the corresponding parent and child transactions \( (p, c) \) are either both included in or both excluded from a chunk.
A slack variable with \( d_{p,c} > 0 \) indicates that the parent is included in the current best closure, but the child is not.
Additionally, all chunks, except for one, %
contain a free transaction variable, and are thus %
excluded from the current best closure.

Overall, this means that each BFS of the LP corresponds to a unique partitioning of the %
DAG into chunks. Each chunk internally forms a spanning tree.
Among these chunks, exactly one (the one with no free transaction variable) 
attains the highest fee rate possible.

Each pivot operation in the simplex algorithm corresponds to a structural modification of the chunks:

\vspace{\baselineskip}
\noindent\emph{1. A free variable is selected to enter the basis:}
This choice is guided by seeking to improve the objective function (overall fee rate).
\begin{itemize}
    \item \emph{If a transaction variable $x_i$ is chosen to enter:}
    An entire chunk that was previously excluded (because $x_i$ was free, setting all its chunk transaction values to $0$) is now targeted. This previously-excluded chunk has a fee rate higher than that of the current `included' chunk (the current best solution candidate). It may become included, or end up being merged with another (included or excluded) chunk.

    \item \emph{If a slack variable $d_j$ is chosen to enter:}
    A $d_j$ variable that was free (which enforced $x_{p_j}=x_{c_j}$ for its associated parent $p_j$ and child $c_j$, keeping them in the same chunk) becomes basic. This effectively splits the original chunk into two separate sub-chunks: one containing \( p_j \), and the other \( c_j \). Such a pivot may be selected if the sub-chunk containing \( p_j \) has a higher fee rate than the currently included chunk, or if the one containing \( c_j \) has a lower fee rate than the currently included chunk.
\end{itemize}

\vspace{\baselineskip}
\noindent\emph{2. A basic variable is selected to leave the basis (becoming free):}
This step maintains the correct number of free and basic variables and adjusts the solution structure.
\begin{itemize}
    \item \emph{If a transaction variable $x_i$ is chosen to leave:}
    The $x_i$ variable associated with the previously 
    `included' chunk becomes free. This forces its value to $0$, thereby excluding that entire chunk from the solution. This occurs if the chunk that was ``made non-excluded'', or split off as described above, has no (basic) dependencies on another chunk.

    \item \emph{If a dependency slack variable $d_j$ is chosen to leave:}
    This enforces $x_{p_j}=x_{c_j}$ for the transactions $p_j$ and $c_j$ linked by this $d_j$. The effect is to merge the ``non-excluded'' or newly split-off chunk with an adjacent chunk it depends on.
\end{itemize}

Through this process, the simplex algorithm systematically explores the space of valid chunk decompositions, moving from one vertex of the feasible polytope to another, until it identifies the dependency-respecting subset of transactions with the highest fee rate.  

Based on the understanding established in this section, the next section introduces an alternative approach to solving the mempool linearization problem.

\section{%
Spanning Forest Linearization}
\label{s:approach-spl}

This section introduces the SFL algorithm, which computes a linearization by iteratively refining a partition of the transaction graph \( G = (V, E) \) into disjoint \emph{chunks}. Inspired by the chunking behavior of BFS in the simplex method, SFL begins with all transactions as singletons and proceeds through a sequence of merge and split operations guided by local fee rate comparisons.
The algorithm first produces an optimal linearization and subsequently performs a minimization phase to obtain
a minimal optimal linearization. The parent-to-child representation is used throughout this section.

\subsection{Basic SFL Algorithm}
\label{ss:spl-algorithm}

The basic algorithm maintains a subset of ``active'' dependencies in $E$. Each active dependency
corresponds to a free slack variable from the %
simplex formulation.
A crucial invariant of the algorithm is that the set of active dependencies must remain acyclic when their directions are ignored.
\begin{definition}[active forest and tree]
\label{def:active-subgraph}
    Let $%
    (V, E)$ denote a transaction dependency graph. %
    Let $A \subseteq E$ denote a set of active dependencies such that $(V,A)$ contains no undirected cycles (i.e., its underlying undirected graph is a forest). Let $\{T_1,\dots,T_m\}$ denote the set of weakly connected components of $(V,A)$; each $T_i$ is a polytree, i.e., its underlying undirected graph is a tree. We refer to $T_i$ as an \emph{active tree} (or, with slight abuse of terminology, simply as a tree). We call  $%
    (V, A)$ the \emph{active forest} and write $T\in(V,A)$ to denote that $T$ is one of the active trees of $(V,A)$, and write $v\in T$ to denote that transaction $v$ belongs to active tree $T$.
\end{definition}

\begin{definition}[parent and child trees]
\label{def:parent-child-trees}
    Let $%
    (V, A)$ be the active forest of a mempool $(V, E, f, s)$. Let $A'=A\setminus\{(p,c)\}$, so $(p,c)$ is inactive in $(V,A')$. We use $H_p$ to denote the active tree in $(V,A')$ that includes $p$, called the \emph{parent tree}, and use $H_c$ to denote the active tree in $(V,A')$ that includes $c$, called the \emph{child tree}.
\end{definition}

\begin{algorithm}[t]
\caption{Basic SFL} %
\begin{algorithmic}[1]
\STATE \textbf{Input:} A mempool $(V,E,f,s)$
\STATE \textbf{Initialize:} 
Set of active dependencies $A \gets \emptyset$
\REPEAT
    \STATE $\mathcal{E}_{\text{split}} \gets \{ (p, c) \in A \mid {}$
    $H_p
    \succ H_c
    \}$ \label{alg-line:sfl-deactivate-dep}
    \STATE $\mathcal{E}_{\text{merge}} \gets \{ (p, c) \in E \setminus A \mid 
    H_c \succeq H_p \}$ \label{alg-line:sfl-merge-condition}
    \STATE Pick arbitrary $e \in \mathcal{E}_{\text{merge}} \cup \mathcal{E}_{\text{split}}$
    \STATE %
    Add $e$ to $A$ if $e \in \mathcal{E}_{\text{merge}}$; %
    remove $e$ from $A$ if $e \in \mathcal{E}_{\text{split}}$
\UNTIL{$\mathcal{E}_{\text{merge}} \cup \mathcal{E}_{\text{split}} = \emptyset$}
\STATE $(T_1, \dots, T_m) \gets \text{
active trees of } (V, A)$,
sorted by descending fee rate, %
with ties broken arbitrarily
\STATE $\sigma%
\gets$ a valid ordering induced by $(T_1,\dots,T_m)$
\STATE \textbf{Return:} \( %
(T_1,\dots,T_m), \sigma%
\)
\end{algorithmic}
\label{alg:spl-basic}
\end{algorithm}

The basic SFL outlined in Algorithm~\ref{alg:spl-basic}
proceeds by repeatedly applying %
two %
operations: %
\begin{itemize}
    \item %
    Merge: For any inactive dependency $(p, c)$, %
    if %
    $H_c \succeq H_p$,
    then activate the dependency to merge %
    the parent and child trees %
    into one.

    \item Split: %
    For any active dependency $(p,c)$, %
    if $H_p \succ H_c$, then deactivate the dependency to split the tree
    into two trees.
\end{itemize}
This process continues until no further merges or splits are possible. At that point, the set of chunks corresponding to the final forest of trees is sorted in descending order of fee rate to produce the linearization. 
If two chunks have equal fee rates,
they can be ordered arbitrarily.
\begin{lemma}[boundary partition]
\label{lem:boundary-partition}
    In an active tree $(H, E_H)$, let $I$ denote a closure that is a non-empty proper subset of \( H \), and let $J = H \setminus I$. Let $B = \{(p,c) \in E_H \mid p \in I,\, c \in J\}$.
    Then the sets $\{H_c\}_{(p,c) \in B}$ are pairwise disjoint and
    \begin{align}
    \label{eq:boundary-part}
        J = \bigcup_{(p,c) \in B} H_c .
    \end{align}
\end{lemma}

\begin{proof}
Since $I$ is a closure in $(H, E_H)$, no edge of $E_H$ enters $I$ from $J$; hence every edge with one endpoint in $I$ and one in $J$ lies in $B$, and
removing the $|B|$ edges in $B$ 
severs all paths between $I$ and $J$, partitioning $H$ into exactly $|B| + 1$ disjoint connected components. Since $I$ remains exactly one of these components, the remaining $|B|$ components are precisely the child trees $H_c$ for each $(p,c) \in B$, establishing the disjoint union~\eqref{eq:boundary-part}.
\end{proof}

\begin{definition}[$q$-function]
\label{def:q-function}
For any two 
sets of transactions, $A$ and $B$, we define the $q$-function as:
\begin{align}
    q(A, B) \triangleq f(A)s(B) - f(B)s(A).
\end{align}
\end{definition}

It is easy to see that the $q$-function is antisymmetric, i.e.,
$q(A, B) = -q(B, A)$ for all $A$ and $B$, which implies that $q(A,A) = 0$. Moreover,
for non-empty $A$ and $B$, %
$q(A, B) > 0$ if and only if $r(A) > r(B)$.

\begin{lemma}[bilinearity of $q$-function]
\label{lem:q-bilinearity}
The $q$-function is bilinear over disjoint unions, i.e.,
\begin{align}
    q(A \cup B, H) &= q(A, H) + q(B, H) \label{eq:qabh} \\
    q(H, A \cup B) &= q(H, A) + q(H, B) \label{eq:qhab}
\end{align}
for all $A$, $B$, and $H$ satisfying $A\cap B=\emptyset$. In particular,
\begin{align}
    q(A\cup B, B) = q(A,B). \label{eq:qabb}
\end{align}
\end{lemma}

\begin{proof}
Because $A$ and $B$ are disjoint, we have $f(A \cup B) = f(A) + f(B)$ and $s(A \cup B) = s(A) + s(B)$. This results in:
\begin{align}
    q(A \cup B, H)
    &= f(A \cup B)s(H) - f(H) s(A \cup B) \\
    &= \bigl(f(A) + f(B)\bigr) s(H) - f(H) \bigl(s(A) + s(B)\bigr) \\
    &= q(A,H) + q(B,H).
\end{align}
Hence~\eqref{eq:qabh} is established, and~\eqref{eq:qhab} is then easily proved using the antisymmetric property of $q$. Applying~\eqref{eq:qabh} with $H=B$ and using the fact that $q(B,B)=0$ lead to~\eqref{eq:qabb}.
\end{proof}

\begin{lemma}[AUC in terms of $q$-values]
\label{lem:auc-q-form}
    For any partition $U^m = (U_1, \dots, U_m)$ of $V$,
    \begin{align}
        A(U^m) = \frac{1}{2} f(V)\, s(V)
            + \frac{1}{2} \sum_{1 \le l < k \le m} q(U_l, U_k). \label{eq:auc-q-statement}
    \end{align}
\end{lemma}
\begin{proof}
    Throughout this proof, $\sum_{l<k}$ abbreviates $\sum_{1 \le l < k \le m}$ and $\sum_{l \neq k}$ abbreviates $\sum_{1 \le l,k \le m,\, l \neq k}$. Expanding the definition of the AUC~\eqref{eq:AUmsum}, we have
    \begin{align}
        A(U^m)
        &= \sum_{l < k} f(U_l) s(U_k)
            + \frac{1}{2} \sum_{k=1}^m f(U_k) s(U_k). \label{eq:auc-q-expand} \\
        &= \frac12 \sum_{l<k} \big(f(U_l)s(U_k)+f(U_k)s(U_l) \notag \\
          &  \qquad\qquad +q(U_l,U_k) \big)
          + \frac12 \sum_{k=1}^m f(U_k) s(U_k) \label{eq:AUmffqf} \\
        &= \frac12 \sum^m_{l=1}\sum^m_{k=1} f(U_l)s(U_k)
        + \frac12 \sum_{l<k} q(U_l,U_k) \label{eq:AUmfq}
    \end{align}
    where~\eqref{eq:AUmffqf} is due to Definition~\ref{def:q-function}, and rearranging the terms yields~\eqref{eq:AUmfq}, which becomes~\eqref{eq:auc-q-statement}.
\end{proof}

Lemma~\ref{lem:auc-q-form} sheds further light on the relationship between AUC and the valid ordering objective in~\eqref{eq:expected-income-expanded}.
By the bilinearity of $q$ (Lemma~\ref{lem:q-bilinearity}),
the sum in~\eqref{eq:AUmfq} can be expressed as
\begin{align}
    \sum_{l < k} q(U_l, U_k) = \sum_{l < k} \;\sum_{i \in U_l,\, j \in U_k} q(\{i\}, \{j\})
    ,
\end{align}
which shows that the AUC depends only on pairs of transactions belonging to different chunks.
In contrast, for a complete ordering $\sigma$, Appendix~\ref{a:gammaq} shows that the expected fee income~\eqref{eq:expected-income-expanded}, multiplied by the constant $s(V)$, can be written as
\begin{align}
\begin{split}
    & f(V) s(V) - \Gamma(\sigma) \\
    &= \frac12 f(V)s(V) -  \frac12 \sum_{i\in V} f_i s_i
    + \frac12 \sum_{\substack{i,j\in V\\ \sigma(j)>\sigma(i)}} q(\{i\},\{j\})  ,
    \label{eq:gammaq}
\end{split}
\end{align}
which includes a summation over all (ordered) pairs of transactions, including pairs that belong to the same chunk.
The chunk decomposition thus retains exactly the cross-chunk contributions while discarding the intra-chunk ordering terms. These discarded terms constitute the NP-hard residue identified by the Sidney decomposition.

\begin{lemma}[boundary $q$-sum]
\label{lem:q-boundary-sum}
    Let $B$, $H$, $I$, and $J$ be defined as in Lemma~\ref{lem:boundary-partition}. Then
    \begin{align}
    \label{eq:boundary-q-sum-main}
    \sum_{(p,c) \in B} q(H_p, H_c) = q(I, J).
    \end{align}
\end{lemma}

\begin{proof}
For every $(p,c) \in B$, the sets $H_p$ and $H_c$ disjointly partition $H$. By Lemma~\ref{lem:boundary-partition}, the child trees $\{H_c\}$ disjointly partition $J$. %
We repeatedly apply %
Lemma~\ref{lem:q-bilinearity} %
to write %
\begin{align}
    \sum_{(p,c) \in B} q(H_p, H_c) 
    &= \sum_{(p,c) \in B} q(H_p \cup H_c, H_c) \label{eq:bqs-step1} \\
    &= \sum_{(p,c) \in B} q(H, H_c) \\
    &= q\left(H, \bigcup_{(p,c) \in B} H_c\right) \label{eq:bqs-step3} \\
    &= q(H, J) \\
    &= q(I \cup J, J) \\
    &= q(I, J)
\end{align}
where~\eqref{eq:bqs-step1} follows from~\eqref{eq:qabb} of Lemma~\ref{lem:q-bilinearity}, and~\eqref{eq:bqs-step3} uses the fact that the child trees $\{H_c\}_{(p,c) \in B}$ are pairwise disjoint by Lemma~\ref{lem:boundary-partition}.

\end{proof}

\begin{theorem}[SFL correctness]
\label{th:sfl-correctness}
If Algorithm~\ref{alg:spl-basic} terminates, it outputs an %
optimal linearization.
\end{theorem}

\begin{proof}
The proof proceeds in two steps. We first show the algorithm produces a valid linearization per Definition~\ref{def:L} and then prove this linearization is optimal.

Let $U_1,\dots,U_m$ denote the corresponding sets of transactions in the trees $T_1,\dots,T_m$, sorted by fee rate. Suppose, for contradiction, $(U_1, \dots, U_m)$ is not a linearization. Since the chunks form a partition of $V$ and satisfy $r(U_1) \ge \dots \ge r(U_m)$, there must exist a chunk $U_k$ that is \emph{not} a closure in the subgraph induced by $V\setminus B_{k-1}$ where $B_{k-1}=U_1\cup\dots\cup U_{k-1}$.
This implies there must exist %
a transaction $c \in U_k$ and its %
parent transaction $p \in U_j$ %
for some $j > k$. Since $p$ and $c$ are in different chunks, the dependency $(p, c)$ must be \emph{inactive}. Moreover, since the merge condition (line~\ref{alg-line:sfl-merge-condition}) $H_c\succeq H_p$ is satisfied, $(p,c)$ needs to be activated according to Algorithm~\ref{alg:spl-basic}, which contradicts the assumption of termination. Hence the output must be a valid linearization.

We next show that $U_1$ has the highest fee rate among all closures in $B_1 = V$. Since every closure in $(V,E)$ is a closure in $(V,A)$, it suffices to show that $U_1$ is a highest-fee-rate closure in $(V,A)$.

Suppose, for contradiction, that some closure $U \subseteq V$ in $(V,A)$ %
satisfies $r(U) > r(U_1)$. Since $r(U)$ is the weighted average of $\{r(U \cap U_i)\}_{i=1}^m$ and $r(U_1) \ge r(U_i)$ for all $i$, some $k$ satisfies
\begin{align}
\label{eq:intersection-feerate}
r(U \cap U_k) > r(U_1) \ge r(U_k).
\end{align}
Set $I = U \cap U_k$ and $J = U_k \setminus I$. By Lemma~\ref{lm:closures_properties}, $I$ is a closure in $(V,A)$, %
and~\eqref{eq:intersection-feerate} forces $I$ to be a non-empty proper subset of $U_k$. Let $B = \{(p,c) \in E_A \mid p \in I,\ c \in J\}$. From~\eqref{eq:intersection-feerate}, $r(I) > r(U_k) > r(J)$, so by Definition~\ref{def:q-function}, $q(I,J) > 0$. Lemma~\ref{lem:q-boundary-sum} then yields
\begin{align}
\sum_{(p,c) \in B} q(H_p, H_c) \;=\; q(I, J) \;>\; 0,
\end{align}
so some $(p,c) \in B$ has $H_p \succ H_c$, satisfying the split condition (line~\ref{alg-line:sfl-deactivate-dep}) and contradicting termination. Thus $U_1$ is a highest-fee-rate closure in $(V,A)$, hence also in $(V,E)$. %

The same argument applied recursively shows $U_k$ is a highest-fee-rate closure in the subgraph induced by $V\setminus B_{k-1}$ for every $k$, so $(U_1, \dots, U_m)$ is an optimal linearization. 
\end{proof}

In contrast to the simplex method of Sec.~\ref{ss:simplex-alg}, SFL does not track the identity of the free transaction variable within each excluded chunk, as it has no bearing on the chunk decomposition. Nor does SFL have a counterpart to the `included' chunk, since it seeks all successive optimal chunks rather than a single highest-fee-rate closure. Merge and split decisions therefore rely on local fee rate comparisons: the two chunks involved (or, for a split, the two would-be chunks) are compared directly with each other rather than against an included chunk. In a way, SFL adopts the operational logic of the simplex method while discarding details irrelevant to the chunk structure.

The basic SFL algorithm leaves several choices unspecified. In the next subsection, we identify these ambiguities and propose refinements to address them.

\subsection{Initialization}
\label{ss:spl-refinements}

\begin{algorithm}
\caption{Initial Linearization}
\begin{algorithmic}[1]

\STATE \textbf{Procedure:} \textsc{MergeStep}(\( H \), \( d \), \(A\)):
\IF{\( d = \text{up} \)}
    \STATE $\mathcal{C} \gets \{\, (p,c) \in E \setminus A \mid H_c = H,\; H_p \preceq H \,\}$ \label{alg-line:mergestep-C-up}
\ELSE
    \STATE $\mathcal{C} \gets \{\, (p,c) \in E \setminus A \mid H_p = H,\; H_c \succeq H \,\}$ \label{alg-line:mergestep-C-down}
\ENDIF
\IF{\( \mathcal{C} = \emptyset \)}
    \RETURN \( \bot \)
\ENDIF
\IF{\( d = \text{up} \)}
    \STATE $\mathcal{E}_{\text{merge}} \gets \{\, (p,c) \in \mathcal{C} \mid r(H_p) = \min_{(p',c') \in \mathcal{C}} r(H_{p'}) \,\}$
    \STATE \((p,c) \xleftarrow{R} \mathcal{E}_{\text{merge}}\);\quad $H' \gets H_p$
\ELSE
    \STATE $\mathcal{E}_{\text{merge}} \gets \{\, (p,c) \in \mathcal{C} \mid r(H_c) = \max_{(p',c') \in \mathcal{C}} r(H_{c'}) \,\}$
    \STATE \((p,c) \xleftarrow{R} \mathcal{E}_{\text{merge}}\);\quad $H' \gets H_c$
\ENDIF
\STATE Add \((p,c)\) to \(A\)
\RETURN $H \cup H', A$

\vspace{1em}

\STATE \textbf{Procedure:} \textsc{Linearize}(mempool $(V, E, f, s)$, optional valid ordering \( \sigma \))
\STATE $A \gets \emptyset$

\IF{valid ordering \( \sigma = (v_1, \dots, v_n)\) is provided} \label{alg-line:mode1-start}
    \FOR{$i = 1, \dots, n$} \label{alg-line:mode1-for}
        \STATE $H \gets \{T \in 
        (V,A)
        \mid v_i \in 
        T
        \}$
        \STATE $(H, A) \gets \textsc{MergeStep}(H, \text{up}, A)$
        \WHILE{$H \neq \bot$} \label{alg-line:mode1-while-start}
            \STATE $(H, A) \gets \textsc{MergeStep}(H, \text{up}, A)$
        \ENDWHILE \label{alg-line:mode1-while-end}
    \ENDFOR \label{alg-line:mode1-end}
\ELSE \label{alg-line:mode2-start}
    \STATE $Q \gets$ random permutation of 
    active trees of $(V,A)$ \label{alg-line:q-init}
    \REPEAT
        \STATE $H \gets Q.\text{pop}()$ \label{alg-line:q-pop}
        \IF{$H \in 
        (V,A)
        $}
            \STATE $(d_1, d_2) \gets$ random permutation of $\{\text{up}, \text{down}\}$
            \FOR{$i = 1,2$} \label{alg-line:mode2-for-start}
                \STATE $(H, A) \gets \textsc{MergeStep}(H, d_i, A)$
                \IF{$H \neq \bot$}
                    \STATE Push $H$ onto $Q$ \label{alg-line:q-push}
                    \STATE \textbf{break}
                \ENDIF
            \ENDFOR \label{alg-line:mode2-for-end}
        \ENDIF
    \UNTIL{$Q = \emptyset$} \label{alg-line:mode2-end}
\ENDIF

\STATE $(T_1, \dots, T_m) \gets$ sort active trees of
$(V,A)$
by descending fee rate, breaking ties uniformly at random 
\label{alg-line:sort}
\STATE $\sigma_{\text{out}} \gets$ valid ordering induced by $(T_1, \dots, T_m)$
\STATE \textbf{Return:} $(T_1, \dots, T_m),\, \sigma_{\text{out}}, A$
\end{algorithmic}
\label{alg:initial-linearize}
\end{algorithm}

We next describe a procedure for constructing a valid initial state from which the algorithm can begin optimization. The procedure \textsc{Linearize} (Algorithm~\ref{alg:initial-linearize}) starts with all dependencies inactive, so that each transaction forms its own singleton chunk, and then merges chunks until the resulting decomposition forms a valid linearization.

A single merge step (\textsc{MergeStep}) attempts to merge a given chunk $H$ in a specified direction by finding a neighboring chunk connected through an inactive dependency: when merging upward, it selects the lowest-fee-rate chunk that $H$ depends on among those with fee rates lower than or equal to that of $H$; when merging downward, it selects the highest-fee-rate chunk that depends on $H$ among those with fee rates higher than or equal to that of $H$. If multiple candidates have equal fee rates, one is chosen uniformly at random. The merge is performed by activating a uniformly random inactive dependency between the two chunks.

When an existing valid ordering %
is provided, the procedure processes each transaction from first to last in %
that order. For each transaction, it repeatedly applies upward merge steps on its chunk until no further merge is possible. This guarantees that the resulting linearization is at least as good as the input ordering, %
since each merge maintains or improves the AUC.

When no existing valid ordering is available, the procedure processes a randomly-shuffled queue of chunks. Each chunk that is popped from the queue attempts a merge step in each of two randomly ordered directions, stopping at the first successful merge. When a merge succeeds, the resulting chunk is re-enqueued for further processing. If neither direction yields a merge, the chunk is simply discarded from the queue. The procedure terminates when the queue is empty, at which point sorting the chunks by decreasing fee rate yields a valid linearization.

\begin{lemma}[validity of initial linearization]
\label{lem:initlinearize-valid}
    Upon termination, Algorithm~\ref{alg:initial-linearize} outputs a valid linearization.
\end{lemma}

\begin{proof}
Since the algorithm only adds active edges to $A$, it is guaranteed to terminate.
By line~\ref{alg-line:sort}, it suffices to show that at termination,
\begin{align}
    H_p \succeq H_c, \quad \forall (p,c) \in E \setminus A. \label{eq:goal}
\end{align}
Suppose, for contradiction, $\exists (p,c) \in E \setminus A$ with $H_c \succ H_p$ at termination.

\textbf{Mode~1} (lines~\ref{alg-line:mode1-start}--\ref{alg-line:mode1-end}):
Suppose the algorithm calls \textsc{MergeStep} $N$ times before terminating. Index the successive \textsc{MergeStep} calls executed during Mode~1 by $k = 1, \dots, N$, and let $H_p^{(k)}, H_c^{(k)}$ denote parent and child trees of $(p,c)$ \emph{after} the $k$-th call; set $H_p^{(0)}, H_c^{(0)}$ to the singletons $\{p\}, \{c\}$. Thus $H_p^{(N)} = H_p$ and $H_c^{(N)} = H_c$.

Let $k^* \in \{0, 1, \dots, N\}$ be the largest index at which the $k^*$-th call modifies the tree containing $c$ (or $k^* = 0$ if no such call exists). By maximality of $k^*$, $H_c^{(k^*)} = H_c$. We claim
\begin{align}
    H_p^{(k^*)} \succ H_c. \label{eq:m1-base}
\end{align}
Indeed, immediately after call $k^*$ the inner while loop (lines~\ref{alg-line:mode1-while-start}-\ref{alg-line:mode1-while-end}) reattempts $\textsc{MergeStep}(H_c, \mathrm{up}, A)$, which by maximality of $k^*$ leaves $H_c$ unchanged, hence returns $\bot$ with $\mathcal{C} = \emptyset$. Since $(p,c) \in E \setminus A$ at termination (and hence throughout, as $A$ only grows), line~\ref{alg-line:mergestep-C-up} forces $H_p^{(k^*)} \succ H_c$.

We now show $H_p^{(k)} \succeq H_p^{(k^*)}$ for all $k \geq k^*$ by induction on $k$. The base $k = k^*$ is trivial. For the inductive step $k \to k+1$, if call $k+1$ does not modify the tree containing $p$, then $H_p^{(k+1)} = H_p^{(k)}$ and we are done. Otherwise $H_p^{(k+1)} = H_p^{(k)} \cup \bar H$, where $\bar H$ is the other tree involved in the $(k+1)$-th \textsc{MergeStep} call. We split into two cases:
\begin{itemize}
    \item $H_p^{(k)}$ is the selected parent in the call: line~\ref{alg-line:mergestep-C-up} gives $\bar H \succeq H_p^{(k)} \succeq H_p^{(k^*)}$.
    \item $H_p^{(k)}$ is the initiator: line~\ref{alg-line:mergestep-C-up} gives $\bar H \preceq H_p^{(k)}$, where $\bar H$ is an inactive parent of $H_p^{(k)}$. We claim $\bar H \succeq H_p^{(k^*)}$. Otherwise, let $k_0>k^*$ be the first call after $k^*$ that modifies the tree containing $p$ by an upward merge with an inactive parent $\bar H_0$ satisfying $\bar H_0 \prec H_p^{(k^*)}$. By minimality of $k_0$ and the induction hypothesis, $H_p^{(k_0-1)} \succeq H_p^{(k^*)}$. Hence $\bar H_0 \prec H_p^{(k_0-1)}$, so $\bar H_0$ belongs to the upward candidate set $\mathcal{C}$ of $H_p^{(k_0-1)}$. But when $H_p^{(k_0-1)}$ was formed, Mode~1 exhausted upward merges from it, so no such inactive parent could remain. This contradiction proves $\bar H \succeq H_p^{(k^*)}$.

\end{itemize}
In both cases $\bar H \succeq H_p^{(k^*)}$, and combined with $H_p^{(k)} \succeq H_p^{(k^*)}$, we get $H_p^{(k+1)} \succeq H_p^{(k^*)}$.
Therefore $H_p = H_p^{(N)} \succeq H_p^{(k^*)} \succ H_c$ by~\eqref{eq:m1-base}, contradicting $H_c \succ H_p$.

\textbf{Mode~2} (lines~\ref{alg-line:mode2-start}--\ref{alg-line:mode2-end}):
Let $H_\ell \in \{H_p, H_c\}$ be whichever is formed last (line~\ref{alg-line:q-init} or~\ref{alg-line:q-push}). $H_\ell$ is enqueued at formation and never subsequently modified, so it is eventually popped (line~\ref{alg-line:q-pop}) with $H_\ell \in 
(V,A)$ and both directions are attempted (lines~\ref{alg-line:mode2-for-start}--\ref{alg-line:mode2-for-end}). 
The dependency $(p,c) \in E \setminus A$ satisfies:
\begin{itemize}
    \item direction up on  $H_c: (p,c) \in \mathcal{C} \text{ since } H_p \prec H_c$,
    \item direction down on  $H_p: (p,c) \in \mathcal{C} \text{ since } H_c \succ H_p$.
\end{itemize}

So whichever direction matches $H_\ell$, $\mathcal{C} \neq \emptyset$ and \textsc{MergeStep} succeeds, modifying $H_\ell$, a contradiction.
\end{proof}

\subsection{Refining the Basic SFL Algorithm}

Until now, we have left unspecified the strategy for selecting which dependencies to activate or deactivate.
The building blocks for the refined SFL algorithm, including the \textsc{MergeUpwards}, \textsc{MergeDownwards}, and \textsc{Improve} procedures, are presented in Algorithm~\ref{alg:sfl-refined}.

Although it appears that making arbitrary valid choices can still yield an optimal result, such random selection may lead to revisiting previous states and performing redundant work. To avoid this inefficiency, we will introduce some heuristics and strategies to guide the algorithm and improve its convergence.

One simple yet effective refinement is to prioritize merging over splitting. Conceptually, merging serves to enforce 
dependency constraints
by joining chunks, whereas splitting serves to improve the overall fee rate by refining chunk structure. By ensuring that, after each improving split, all possible merges are performed immediately, the algorithm quickly restores a 
valid linearization
state. This prioritization has a practical benefit: if the algorithm is interrupted early (e.g., due to a time limit), the current state remains a valid linearization, even if it is not optimal.

\begin{algorithm}
\caption{Refined SFL Building Blocks}
\begin{algorithmic}[1]

\STATE \textbf{Procedure:} \textsc{MergeUpwards}($t$, $A$):
\STATE $H \gets$ active tree in $(V,A)$ containing $t$
\REPEAT
    \STATE $(H, A) \gets \textsc{MergeStep}(H, \text{up}, A)$
\UNTIL{$H = \bot$}
\RETURN $A$

\vspace{0.5em}

\STATE \textbf{Procedure:} \textsc{MergeDownwards}($t$, $A$):
\STATE $H \gets$ active tree in $(V,A)$ containing $t$
\REPEAT
    \STATE $(H, A) \gets \textsc{MergeStep}(H, \text{down}, A)$
\UNTIL{$H = \bot$}
\RETURN $A$

\vspace{0.5em}

\STATE \textbf{Procedure:} \textsc{Improve}($(p,c)$, $A$): \label{alg-line:improve1}
\IF{$H_c \succ H_p$}
    \RETURN $A$
\ENDIF
\STATE $\mathcal{E}_{\text{remerge}} \gets \{(u,v) \in E \setminus A \mid u \in H_c,\; v \in H_p\}$ \label{alg-line:improve3}
\STATE $A \gets A \setminus \{(p,c)\}$ \label{alg-line:improve2}
\IF{$\mathcal{E}_{\text{remerge}} \neq \emptyset$}
    \STATE Pick arbitrary $e \in \mathcal{E}_{\text{remerge}}$
    \STATE $A \gets A \cup \{e\}$ \label{alg-line:improve4}
\ELSE \label{alg-line:improve5}
    \STATE $A \gets \textsc{MergeUpwards}(p, A)$ \label{alg-line:improve6}
    \STATE $A \gets \textsc{MergeDownwards}(c, A)$ \label{alg-line:improve7}
\ENDIF \label{alg-line:improve8}
\RETURN $A$
\end{algorithmic}
\label{alg:sfl-refined}
\end{algorithm}

Let us define an \emph{improvement step} as one application of the splitting rule (selected by any criterion), followed by all possible merges.
Let $(U_1, \dots, U_m)$ be a valid linearization, and let $(p,c)$ be an active dependency within some chunk $U_k$ whose deactivation yields a parent-side chunk $H_p$ and a child-side chunk $H_c$ with $r(H_p) > r(H_c)$.

\begin{definition}[improvement step]
\label{def:improvement-step}
    An \emph{improvement step} on $(p,c)$ follows the logic of the \textsc{Improve} procedure (Algorithm~\ref{alg:sfl-refined}, lines~\ref{alg-line:improve1}--\ref{alg-line:improve8}):
    \begin{enumerate}
        \item \textbf{Split:} 
        Deactivate the dependency $(p, c)$, splitting $U_k$ into $H_p$ and $H_c$ (line~\ref{alg-line:improve2}).
        
        \item \textbf{Resolution:} Proceed with one of the following operations based on dependency constraints:
        \begin{enumerate}
            \item \textbf{Immediate Re-merge:} If an inactive dependency exists such that the parent tree $H_p$ depends on the child tree $H_c$, activate it (lines~\ref{alg-line:improve3}--\ref{alg-line:improve4}). This results in a \emph{self-merge} where the linearization remains unchanged. 
            \item \textbf{Restoration:} Otherwise, apply \textsc{MergeUpwards} to $H_p$ and \textsc{MergeDownwards} to $H_c$ until no further merges are possible (lines~\ref{alg-line:improve6}--\ref{alg-line:improve7}). Sorting the new set of chunks results in a new valid linearization. 
        \end{enumerate}
    \end{enumerate}
\end{definition}

When a split occurs in the improvement step, two possibilities arise: 

1) If the higher-fee-rate chunk (the parent side), through another inactive dependency, depends on the lower-fee-rate chunk (the child side), the 
rule causes them to re-merge. The chunk %
remains the same as before the split, albeit with a different internal active tree. Whether this behavior can lead to infinite loops, or how many times it can repeat, is a concern that is addressed in Secs.~\ref{ss:sfl-termination-bland} and~\ref{ss:randomized-sfl}, where we introduce specific rules for selecting which dependency to split so that the algorithm is guaranteed to terminate.

2) If no such self-merge happens (because the split-off parent chunk does not depend on the split-off child), the parent and child chunks initially occupy adjacent positions in the ordered list, which may no longer be properly sorted. With a higher fee rate, the parent can only depend on chunks that preceded it, but it may have a higher fee rate than them. The 
rule now reduces to the chunk ``bubbling up'' until it either reaches a position where its fee rate is no longer higher than the preceding chunk or it encounters a lower-fee-rate dependency, merges with it, and continues bubbling upward. A reversed process occurs for the split-off child. Its fee rate is now lower than the original chunk’s, so it may be surpassed by successors. The child chunk ``bubbles down'' until it reaches a position where its fee rate exceeds that of the next chunk, or it encounters a higher-fee-rate dependent, merges with it, and possibly continues downward.

\begin{lemma}[area gain from split or merge]
\label{lem:structural-gain}
    Let $S = (U_1, \dots, U_m)$ be a 
    partition of $V$.
    Let 
    \begin{align}
        S' = (U_1, \dots, U_{k-1},\, U_k',\, U_k'',\, U_{k+1}, \dots, U_m)
    \end{align}
    where $U_k' \cup U_k''=U_k$ is a partition of $U_k$ into two disjoint, non-empty subsets. Then:
    \begin{equation}
        A(S') - A(S) = \frac{1}{2}\, q(U_k', U_k'').
    \end{equation}
    Consequently, splitting $U_k$ into $(U_k', U_k'')$ where 
    $U_k' \succ U_k''$
    yields a strict increase in AUC of $\frac{1}{2}\, q(U_k', U_k'') > 0$. Conversely, given a sequence containing the adjacent pair $(U_k', U_k'')$ where 
    $U_k'' \succ U_k'$, 
    merging them into a single set $U_k = U_k' \cup U_k''$ yields a strict increase in AUC of $\frac{1}{2}\, q(U_k'', U_k') > 0$. If $r(U_k'') = r(U_k')$, merging them yields no change in the AUC.
\end{lemma}

\begin{proof}
    Since all sets except for $U_k$, $U_k'$, and $U_k''$
    are identical in $S$ and $S'$, and share the same cumulative fees from their predecessors, we have
    \begin{align} \label{eq:AS'-AS}
    \begin{split}
        A(S') -& A(S)
        = s(U_k') \!\left( f(U^{k-1}) + \tfrac{1}{2} f(U_k') \right) \\ 
        &+ s(U_k'') \!\left( f(U^{k-1}) + f(U_k') + \tfrac{1}{2} f(U_k'') \right) \\
        &\quad - s(U_k) \!\left( f(U^{k-1}) + \tfrac{1}{2} f(U_k) \right). 
    \end{split}
    \end{align}
    Substituting $s(U_k) = s(U_k') + s(U_k'')$ and $f(U_k) = f(U_k') + f(U_k'')$ into~\eqref{eq:AS'-AS} yields
    \begin{align}
        &A(S') - A(S) \notag \\
        &= s(U_k'') f(U_k') - \tfrac{1}{2} s(U_k') f(U_k'') - \tfrac{1}{2} s(U_k'') f(U_k') \\
        &= \frac{1}{2} \bigl( f(U_k') s(U_k'') - f(U_k'') s(U_k') \bigr) \\
        &= \frac{1}{2}\, q(U_k', U_k'').
    \end{align}
    For a split where 
    $U_k' \succ U_k''$, 
    according to Definition~\ref{def:q-function}, it holds that $q(U_k', U_k'') > 0$, hence the AUC increases by $\frac{1}{2}\,q(U_k', U_k'') > 0$. For a merge where 
    $U_k'' \succ U_k'$, 
    we transition from the sequence $S'$ to $S$, so the gain is $A(S) - A(S') = -\frac{1}{2}\,q(U_k', U_k'') = \frac{1}{2}\,q(U_k'', U_k') > 0$. If $r(U_k'') = r(U_k')$, merging them results in no change in the AUC since  $\frac{1}{2}\,q(U_k'', U_k') = 0$.
\end{proof}

\begin{lemma}[area gain from swapping]
\label{lem:swap-gain}
    Let $S = (U_1, \dots, U_m)$ be a partition of $V$. Suppose $U_{k+1} \succ U_k$ for some $k \in \{1, \dots, m-1\}$. Let $S'$ be the sequence obtained by swapping $U_k$ and $U_{k+1}$:
    \begin{align}
        S' = (U_1, \dots, U_{k-1},\, U_{k+1},\, U_k,\, U_{k+2}, \dots, U_m).
    \end{align}
    Then, the AUC strictly increases by:
    \begin{equation}
        A(S') - A(S) = q(U_{k+1}, U_k) > 0.
    \end{equation}
\end{lemma}

\begin{proof}
    We prove this result by using Lemma~\ref{lem:structural-gain} twice. Specifically, $S'$ is obtained from $S$ by first merging $U_k$ and $U_{k+1}$ to obtain $R$ and then splitting their union in $R$ to $U_{k+1}$ followed by $U_k$. We have
    \begin{align}
        A(S') - A(S)
        &= (A(S')-A(R)) - (A(S)-A(R)) \\
        &= \frac12 q(U_{k+1},U_k) - \frac12 q(U_k,U_{k+1}) \\
        &= q(U_{k+1},U_k)
    \end{align}
    where we have also used the antisymmetry of $q(\cdot,\cdot)$ by Definition~\ref{def:q-function}.
\end{proof}

\begin{theorem}[strict improvement]
\label{thm:strict-improvement}
    Let $L$ be a valid linearization. Let $L'$ be the linearization resulting from an improvement step on a chunk $U \in L$. If a self-merge does not occur (i.e., $L' \neq L$), then the AUC strictly increases:
    \begin{equation}
    \label{eq:theorem-strict}
        A(L') > A(L).
    \end{equation}
\end{theorem}

\begin{proof}
   The improvement step consists of a split phase followed by a resolution phase.
    
    First, in the split phase, the chunk $U$ is replaced by the subsequence $(H_p, H_c)$ with 
    $H_p \succ H_c$. 
    By Lemma~\ref{lem:structural-gain}, this operation strictly increases the AUC by $\frac{1}{2}q(H_p, H_c) > 0$.
    
    Second, in the resolution phase (specifically the restoration case, as no self-merge occurred), the algorithm iteratively restores the sorted order of closures. This process involves examining adjacent chunks $(A, B)$ that violate the sorting condition (i.e., 
    $B \succ A$
    ) and applying one of two operations:
    \begin{itemize}
        \item \textbf{Swapping:} If $A$ and $B$ are independent (neither depends on the other), they are swapped to $(B, A)$. By Lemma~\ref{lem:swap-gain}, this strictly increases the AUC by $q(B, A) > 0$.
        \item \textbf{Merging:} If they cannot be swapped because $A$ is a parent of $B$, they are merged into $M = A \cup B$. By Lemma~\ref{lem:structural-gain}, this 
        increases the AUC by $\frac{1}{2}q(B, A) \ge 0$ (strictly positive if $r(B) > r(A)$, or zero if $r(B) = r(A)$).
    \end{itemize}
    
    Since the initial split yields a strictly positive gain, and every operation in the resolution phase yields a 
    non-negative gain, the total change in AUC is strictly positive.
\end{proof}

Whenever a self-merge does not occur, the fee rate diagram strictly improves: a high-fee-rate chunk moves up, and a low-fee-rate one moves down. Thus, the 
barrier to guaranteed termination is the possibility of indefinite self-merging following a split.

When multiple split options are available during an improvement step, a natural question is how to choose among them. 
Several strategies are possible. One deterministic heuristic, the \emph{maximum-$q$} rule, selects the split maximizing $q$, which corresponds to the largest derivative of the LP objective with respect to the corresponding slack variable. While this heuristic tends to reduce the number of iterations in typical cases, it does not prevent cycling and can worsen adversarial-case performance; we discuss it further in Appendix~\ref{s:sfl-max-q}. In Secs.~\ref{ss:sfl-termination-bland} and~\ref{ss:randomized-sfl}, we introduce alternative split selection rules that provide formal termination guarantees.

\subsection{SFL Algorithm with Guaranteed Termination}
\label{ss:sfl-termination-bland}

\begin{algorithm}
\caption{SFL with Guaranteed Termination}
\label{alg:sfl-deterministic}
\begin{algorithmic}[1]
\STATE \textbf{Procedure:} \textsc{Improve}($(p,c)$, $A$, $\idx$):
\STATE $\mathcal{E}_{\text{remerge}} \gets \{(u,v) \in E \setminus A \mid u \in H_c,\; v \in H_p\}$ \label{alg-line:det-remerge-set}
\STATE $A \gets A \setminus \{(p,c)\}$ \label{alg-line:det-deactivate}
\IF{$\mathcal{E}_{\text{remerge}} \neq \emptyset$}
    \STATE $A \gets A \cup \{\arg\min_{e' \in \mathcal{E}_{\text{remerge}}} \idx(e')\}$ \label{alg-line:merge-rule3}
\ELSE
    \STATE $A \gets \textsc{MergeUpwards}(p, A)$ \label{alg-line:merge-rule1}
    \STATE $A \gets \textsc{MergeDownwards}(c, A)$ \label{alg-line:merge-rule2}
\ENDIF
\RETURN $A$

\vspace{0.5em}
\STATE \textbf{Procedure:} \textsc{DeterministicSFL}(Mempool $(V,E,f,s)$, optional valid ordering $\sigma$):
\STATE Assign unique indices $\iota(e) \in \{1,\dots,|E|\}$ for $\forall e \in E$ \label{alg-line:det-indices}
\STATE $(\cdot,\, \cdot,\, A) \gets \textsc{Linearize}((V,E,f,s),\, \sigma)$ \label{alg-line:det-linearize}
\REPEAT
    \STATE $\mathcal{E}_{\text{split}} \gets \{(p,c) \in A \mid H_p \succ H_c\}$ \label{alg-line:det-split-set}
    \IF{$\mathcal{E}_{\text{split}} \neq \emptyset$}
        \STATE $A \gets \textsc{Improve}(\arg\min_{e \in \mathcal{E}_{\text{split}}} \idx(e),\, A, \idx)$ \label{alg-line:split-rule}
    \ENDIF
\UNTIL{$\mathcal{E}_{\text{split}} = \emptyset$} \label{alg-line:det-until}
\STATE $(T_1, \dots, T_m) \gets$ active trees of $(V,A)$, sorted by descending fee rate, ties broken arbitrarily \label{alg-line:det-sort}
\STATE $\sigma_{\text{out}} \gets$ a valid ordering induced by $(T_1, \dots, T_m)$
\STATE \textbf{Return:} $(T_1, \dots, T_m),\, \sigma_{\text{out}}, A$
\end{algorithmic}
\end{algorithm}

An algorithm %
constructed from the operations in Algorithm~\ref{alg:sfl-refined} may pivot indefinitely among %
equivalent suboptimal forests without meaningful %
improvement in the AUC. To prevent such cycling, we adopt a tie-breaking strategy inspired by Bland's rule for the simplex algorithm~\cite{bland1977new}. The rule enforces a strict lexicographic ordering: whenever multiple valid candidates exist for a self-merge or split operation, the algorithm selects the edge with the \emph{lowest index}. Specifically, we assign a unique, fixed index $\idx (e)$ to every dependency edge $e\in E$ to enable lexicographical tie-breaking in edge comparisons. We present the complete SFL procedure with this indexing convention in Algorithm~\ref{alg:sfl-deterministic}. In the remainder of this subsection, we prove Algorithm~\ref{alg:sfl-deterministic} is guaranteed to terminate at the optimal linearization. 

For a mempool $(V, E, f, s)$, let $U \subseteq V$, and let $F = \{(u,v) \in E \mid u \in U, v \in U\}$ denote the set of dependencies internal to $U$. Assume $(U, F)$ is weakly connected. An active tree that spans all vertices of $U$ can be represented by 
$T=(U,E_T)$, where $E_T \subseteq F$; we call $T$ an \emph{active tree of $(U, F)$}. %

\begin{definition}[tree flow]
    Let $T=(U,E_T)$ be 
    an active tree of a connected graph $(U,F)$ and 
    $e = (p, c) \in E_T$.
    Let $P_e^T%
    $ %
    denote the parent tree of $(p, c)$ in $T$.
    The flow on $e$ with respect to $T$ is denoted as    
    \begin{align}
        x_e^T = 
        q(P_e^T,U)
    \end{align}
\end{definition}

By Definition~\ref{def:q-function}, $q(P_e^T, U) > 0$ if and only if $P_e^T \succ U$, which is equivalent to $P_e^T \succ U \setminus P_e^T$. Thus, $q(P_e^T, U)$ gauges the ``split potential'' across edge $e$: the tree flow of an edge indicates the result of fee rate comparisons for potentially splitting across it.

Given $S \subseteq U$, %
the boundary sets of $S$ in the active tree $T =(U, E_T)$ 
of $(U,F)$ are denoted as:
    \begin{align}
        \delta_{out}^{T}(S) &= \{ (u,v) \in E_T 
        \mid u \in S, v \notin S \}, \\
        \delta_{in}^{T}(S) &= \{ (u,v) \in E_T
        \mid u \notin S, v \in S \}.
    \end{align}

\begin{lemma}[conservation of potential]
    \label{lem:flow-conservation}
    For any $S \subseteq U$,
    \begin{align}
        q(S, U) = \sum_{e \in \delta_{out}^{T}(S)} x_e^T - \sum_{e \in \delta_{in}^{T}(S)} x_e^T.
    \end{align}
\end{lemma}

\begin{proof}
We first prove the special case $S = \{v\}$ for a single vertex $v \in U$. 
Removing $v$ and its incident edges in $T$ partitions $U \setminus \{v\}$ into disjoint connected components $\{C_1, \dots, C_k\}$, where each component $C_i$ connects to $v$ via exactly one edge $e_i \in E_T$. For each $e_i$:
\begin{itemize}
    \item If $e_i \in \delta_{in}^T(\{v\})$, then $P_{e_i}^T = C_i$, so $x_{e_i}^T = q(C_i, U)$.
    \item If $e_i \in \delta_{out}^T(\{v\})$, then $P_{e_i}^T = U \setminus C_i$, so by Lemma~\ref{lem:q-bilinearity},
    \begin{align}
        x_{e_i}^T = q(U \setminus C_i, U) = q(U, U) - q(C_i, U) = -q(C_i, U).
    \end{align}
\end{itemize}
In both cases, the contribution of $e_i$ to the net flow leaving $\{v\}$ equals $-q(C_i, U)$. Hence, using Lemma~\ref{lem:q-bilinearity},
\begin{align}
    \sum_{e \in \delta_{out}^T(\{v\})} x_e^T - \sum_{e \in \delta_{in}^T(\{v\})} x_e^T
    &= -\sum_{i=1}^k q(C_i, U) \\
    &= -q\left(\bigcup_{i=1}^k C_i, U\right) \\
    &= -q(U \setminus \{v\}, U) \\
    &= q(\{v\}, U). \label{eq:flow-cons-single}
\end{align}

For arbitrary $S \subseteq U$, by Lemma~\ref{lem:q-bilinearity} and~\eqref{eq:flow-cons-single},
\begin{align}
    q(S, U)
    &= \sum_{v \in S} q(\{v\}, U) \\
    &= \sum_{v \in S} \left( \sum_{e \in \delta_{out}^T(\{v\})} x_e^T - \sum_{e \in \delta_{in}^T(\{v\})} x_e^T \right) \\
    &= \sum_{\substack{v \in S,\, (v,u) \in E_T}} x_{(v,u)}^T - \sum_{\substack{v \in S,\, (u,v) \in E_T}} x_{(u,v)}^T \label{eq:flow-cons-b4cancel} \\
    &= \sum_{e \in \delta_{out}^T(S)} x_e^T - \sum_{e \in \delta_{in}^T(S)} x_e^T \label{eq:flow-cons-final}
\end{align}
where~\eqref{eq:flow-cons-final} follows from~\eqref{eq:flow-cons-b4cancel} because any edge $e = (u, w) \in E_T$ with both endpoints in $S$ appears once as outgoing (from $u$) and once as incoming (at $w$), so its contributions cancel; only edges with exactly one endpoint in $S$ survive.
\end{proof}

\begin{theorem} \label{thm:terminate}
    The Refined SFL, Algorithm~\ref{alg:sfl-deterministic}, terminates in a finite number of steps.
\end{theorem}

\begin{proof}
    Since all sub-routines' loops in Algorithm~\ref{alg:sfl-deterministic} terminate in at most $|V|$ %
    steps, it suffices to show that the repeat-loop in the main algorithm terminates in finite steps, i.e., \textsc{Improve}(.) in line~\ref{alg-line:split-rule} is called a finite number of times.

    Each improvement step consists of one split (deactivating $e_{out}$) and possibly some %
    merges (activating edges). We distinguish two improvement step types based on their effect on the linearization $L$:

    \textbf{Type 1: Linearization change.} %
    The improvement step results in a new linearization $L' \ne L$. By Theorem~\ref{thm:strict-improvement},
    the 
    total AUC %
    strictly increases. %
    Since the number of possible linearizations is finite, there are at most a finite number of
    type~1 improvement steps. %

    \textbf{Type 2: Basis exchange.} 
    The split of some active tree with vertex set $U$ in the improvement step creates parent and child trees, but a subsequent merge immediately reconnects them via some edge $e_{in}$ without altering $U$. The linearization remains $L$, but the active tree changes. 
    Let $F = \{(u,v) \in E \mid u \in U, v \in U\}$ denote the internal dependencies of $U$, and let $T_t=(U,E_t)$ and $T_{t+1}=(U,E_{t+1})$ denote the active trees of $(U, F)$ 
    at the start of
    improvement steps $t$ and $t+1$, respectively, where we abbreviate $E_{T_t}$ as $E_t$. Then $E_{t+1} = (E_{t} \setminus \{e_{out}\}) \cup \{e_{in}\}$. 
    We similarly write $x_e^{(t)}$ to denote $x_e^{T_t}$, and $\delta_{in}^{(t)}(\cdot)$ and $\delta_{out}^{(t)}(\cdot)$ to denote $\delta_{in}^{T_t}(\cdot)$ and $\delta_{out}^{T_t}(\cdot)$, respectively.

    Assume, for contradiction, that the algorithm enters an infinite cycle of type 2 improvement steps. A cycle is defined as a sequence of improvement steps where the algorithm transitions from an active forest $F_1$ through a sequence of states and returns to $F_1$ without altering the linearization. %

    Let $\mathcal{C} \subseteq F$ be the set of edges that toggle (enter or leave) the active tree of $(U, F)$ during this cycle. Let $e_{\max}$ be the highest-index edge in %
    $\mathcal{C}$. %

    \textbf{Analysis at Improvement Step $t$:} 
    Consider the improvement step $t$ in the cycle where $e_{\max}$ is chosen to split $U$. %
    According to the split rule (line~\ref{alg-line:split-rule}), $e_{\max}$ was chosen because it was the valid split candidate with the \emph{minimum index} and $x_{e_{\max}}^{(t)} > 0$. Therefore, %
    every other toggling edge $f \in \mathcal{C} \cap E_t$ satisfies $\idx (f) < \idx (e_{\max})$ and $x_f^{(t)} \le 0$ (if $x_f^{(t)}>0$, $f$ would have been chosen). 

    \textbf{Analysis at Improvement Step $s$:} 
    Consider the improvement step $s > t$ in the cycle where, immediately after some edge $g \in 
    E_s$ is split, dividing $U$ into $S=H_p^{(s)}$ (parent-side) and $H_c^{(s)}$ (child-side),
    $e_{\max}$ is chosen to merge %
    them back into $U$. Since $g$ was a valid split, 
    $q(S, U) = x_g^{(s)} > 0$.
    According to the remerge rule (line~\ref{alg-line:merge-rule3}), $e_{\max}$ is chosen because it is the candidate with the \emph{minimum index} crossing from $H_c^{(s)}$ to $H_p^{(s)}$. Therefore, no other toggling edge $f \in \mathcal{C} 
    \setminus 
    E_s
    $ crosses the boundary of $
    H_p^{(s)}$ in the merging direction (otherwise its lower index would trigger selection over $e_{\max}$).

    \textbf{The Contradiction:}
    We evaluate $q(S, U)$ for the set $S = H_p^{(s)}$ using the flows from tree $T_t
    $. By Lemma~\ref{lem:flow-conservation}, $q(S, U)$ equals the net flow out of $S$ in $T_t$:
    \begin{align} \label{eq:DUSx-x}
         q(S, U) = \sum_{e \in \delta_{out}^{(t)}(S)} x_e^{(t)} - \sum_{e \in \delta_{in}^{(t)}(S)} x_e^{(t)}.
    \end{align}
    We know 
    $q(S, U) = x_g^{(s)} > 0$. We examine the edges contributing to the right-hand side of~\eqref{eq:DUSx-x}. These must be edges in 
    $E_t$
    that cross $S$. Since 
    $E_s$
    and 
    $E_t$
    differ only by edges in $\mathcal{C}$ (i.e., edges that toggle), and the only edge in 
    $E_s$
    crossing $S$ is $g$, any edge in 
    $E_t$
    crossing $S$ must be either $g$ (if $g \in 
    E_t$) or an edge in $\mathcal{C}$. 
    All edges that contribute to~\eqref{eq:DUSx-x} must be in $
    E_t
    \cap \mathcal{C}$ %
    and cross $S$. They must be in one of the following three groups:
    
    1. %
    Edge $e_{\max} \in 
    E_t
    \cap \mathcal{C}$. At improvement step $s$, we identified that $e_{\max}=(u,v)$ connects $u \in H_c^{(s)}$ to $v \in H_p^{(s)}=S$. Thus, %
    $e_{\max} \in \delta_{in}^{(t)}(S)$. Its contribution to~\eqref{eq:DUSx-x} %
    is $-x_{e_{\max}}^{(t)} < 0$. 

    2. Edge $g$ contributes to~\eqref{eq:DUSx-x} if $g \in 
    E_t$. %
    As $g$ toggles, we have %
    $g \in \mathcal{C}$, %
    so $\idx (g) < \idx (e_{\max})$. By the logic at improvement step $t$, $x_g^{(t)} \le 0$. 
    Since $g\in\delta_{out}^{(t)}(S)$, %
    its contribution to~\eqref{eq:DUSx-x} is not positive.

    3. Any other edge $f \in 
    E_t
    \cap \mathcal{C}$ (with $f \neq e_{\max}$ and $f \neq g$): 
    Since $f$ crosses $S$ and $g$ is the only edge of $E_s$ crossing $S$, $f \notin E_s$. By the merge rule at step $s$, if $f$ also pointed into $S$, it would be a valid merge candidate and would have been chosen over $e_{\max}$ since $\idx(f) < \idx(e_{\max})$; hence $f \in \delta_{out}^{(t)}(S)$.
    We also know $x_f^{(t)} \le 0$. Its contribution to~\eqref{eq:DUSx-x} is $x_f^{(t)} \le 0$. %

    Summing the contributions yield a nonpositive right hand side of~\eqref{eq:DUSx-x}.
    This contradicts the fact that 
    $q(S, U) = x_g^{(s)} > 0$.
    This proves that the maximal-index edge $e_{\max}$ cannot satisfy the conditions to both leave at $t$ and enter at $s > t$ under the minimum-index rule. Consequently, the cycle cannot exist.
\end{proof}

By Theorem~\ref{th:sfl-correctness} and Theorem~\ref{thm:terminate}, the refined 
SFL algorithm (Algorithm~\ref{alg:sfl-deterministic})
is guaranteed to terminate and puts out an optimal linearization.

\subsection{Randomizing the SFL Algorithm}
\label{ss:randomized-sfl}

While the deterministic variant of Sec.~\ref{ss:sfl-termination-bland} provides guaranteed termination, its reliance on fixed tie-breaking rules makes its behavior fully predictable. In order to achieve the desirable algorithmic properties %
discussed in Sec.~\ref{ss:desirable-properties}, especially robustness and fairness, 
we introduce randomization into the SFL algorithm. The key design principle is that every choice point in the algorithm, including which chunks to merge, which dependency to activate or deactivate, and the order in which to process chunks, should incorporate uniform randomness.

\begin{algorithm}
\caption{Randomized SFL}
\label{alg:sfl-randomized}
\begin{algorithmic}[1]
\STATE \textbf{Procedure:} \textsc{Improve}($(p,c)$, $A$):
\STATE $\mathcal{E}_{\text{remerge}} \gets \{(u,v) \in E \setminus A \mid u \in H_c,\; v \in H_p\}$ \label{alg-line:rand-remerge-set}
\STATE $A \gets A \setminus \{(p,c)\}$ \label{alg-line:rand-deactivate}
\IF{$\mathcal{E}_{\text{remerge}} \neq \emptyset$}
    \STATE $e' \xleftarrow{R} \mathcal{E}_{\text{remerge}}$;\quad $A \gets A \cup \{e'\}$ \label{alg-line:rand-merge-rule3}
\ELSE
    \STATE $A \gets \textsc{MergeUpwards}(p, A)$ \label{alg-line:rand-merge-rule1}
    \STATE $A \gets \textsc{MergeDownwards}(c, A)$ \label{alg-line:rand-merge-rule2}
\ENDIF
\RETURN $A$

\vspace{0.5em}
\STATE \textbf{Procedure:} \textsc{RandomizedSFL}(Mempool $(V,E,f,s)$, optional valid ordering $\sigma$):
\STATE $(\cdot,\, \cdot,\, A) \gets \textsc{Linearize}((V,E,f,s),\, \sigma)$ \label{alg-line:rand-linearize}
\STATE $Q \gets$ random permutation of active trees of $(V,A)$ \label{alg-line:rand-queue-init}
\WHILE{$Q \neq \emptyset$}
    \STATE $T \gets$ first active tree in $Q$
    \label{alg-line:rand-pop}
    \IF{$T \in (V,A)$}
        \STATE $\mathcal{E}_{\text{split}} \gets \{(p,c) \in A \mid p,c \in T,\; H_p \succ H_c\}$ \label{alg-line:rand-split-set}
        \IF{$\mathcal{E}_{\text{split}} \neq \emptyset$}
            \STATE $e^* \xleftarrow{R} \mathcal{E}_{\text{split}}$
            \STATE $A \gets \textsc{Improve}(e^*, A)$ \label{alg-line:rand-split-rule}
            \STATE $Q \gets$ random permutation of active trees of $(V,A)$\label{alg-line:rand-queue-reset}
        \ENDIF
    \ENDIF
\ENDWHILE
\STATE $(T_1, \dots, T_m) \gets$ active trees of $(V,A)$, sorted by descending fee rate, ties broken uniformly at random \label{alg-line:rand-sort}
\STATE $\sigma_{\text{out}} \gets$ a valid ordering induced by $(T_1, \dots, T_m)$
\STATE \textbf{Return:} $(T_1, \dots, T_m),\, \sigma_{\text{out}}, A$
\end{algorithmic}
\end{algorithm}

Specifically, the randomized SFL algorithm incorporates the following elements:
\begin{itemize}
    \item \textbf{Random merge target selection:} When merging a chunk upwards or downwards, if multiple candidate chunks have the same fee rate (i.e., are equally valid merge targets), a uniformly random one among them is chosen.
    
    \item \textbf{Random dependency activation:} After selecting two chunks to merge, a uniformly random dependency between them is activated, rather than selecting one by a fixed rule.
    
    \item \textbf{Random chunk processing order:} The queue of chunks to be optimized is initialized in a uniformly random order. This ensures that no particular region of the transaction graph is systematically prioritized over others.
    
    \item \textbf{Random split selection:} When selecting which active dependency to deactivate (split) within a chunk, a uniformly random dependency is chosen among those whose parent fee rate is strictly higher than the child fee rate.
    This ensures that all valid split candidates receive equal attention.

\end{itemize}

Together, these randomized elements ensure that the algorithm avoids predictable behavior and distributes computational work equitably across all transactions, making it resistant to adversarial manipulation. Since non-adversarially-created clusters are almost always linearized optimally in negligible time, the slight overhead of randomization is an acceptable trade-off for robustness. The complete randomized SFL is presented in Algorithm~\ref{alg:sfl-randomized}.

We now prove that the randomized SFL algorithm terminates with probability~1. The proof relies on showing 
that from any state, the deterministic variant (Algorithm~\ref{alg:sfl-deterministic}) provides a finite path to either an unblocked split or termination, and that randomized SFL follows this path with positive probability.

\begin{definition}[forest plateau]
\label{def:forest-plateau}
    Let $L = (U_1, \dots, U_m)$ be a valid linearization of 
    mempool $(V,E,f,s)$. We define the \emph{forest plateau} $\mathcal{F}(L)$ as the set of all valid 
    active forests
    $A \subseteq E$ such that the 
    active trees of
    $(V, A)$ correspond exactly to the chunks $\{U_1, \dots, U_m\}$.
    Any operation that modifies 
    $A$ without changing the linearization $L$ (i.e., a self-merge) constitutes a transition between states within the same plateau $\mathcal{F}(L)$.
\end{definition}

We distinguish two kinds of improvement steps: a \emph{blocked split} occurs when the resolution phase results in an immediate re-merge (leaving the linearization unchanged), and an \emph{unblocked split} occurs otherwise.

\begin{theorem}[almost sure termination]
\label{thm:terminationwp1}
    The randomized SFL algorithm, Algorithm~\ref{alg:sfl-randomized}, terminates with probability 1.
\end{theorem}

\begin{proof}
    The set of possible linearizations of a finite mempool $(V,E,f,s)$ is finite, and by Theorem~\ref{thm:strict-improvement}, each unblocked split strictly increases the AUC. Therefore, at most finitely many unblocked splits can occur over the course of the algorithm. It suffices to show that from any 
    linearization, the algorithm performs only finitely many blocked splits before either terminating or reaching a strictly better linearization.

    Let $L_t$ be the linearization at improvement step $t$. We show that there exists a %
    natural number $h$ %
    such that, with probability~1, within $h$ improvement steps the algorithm either terminates or reaches a new linearization $L_{t+h}$ with $A(L_{t+h}) > A(L_t)$.

    Let $\mathcal{F}(L_t)$ denote the (finite) set of active forests whose 
    active trees correspond to the chunks of~$L_t$ (Definition~\ref{def:forest-plateau}). While the algorithm performs only blocked splits (self-merges), it transitions between states within~$\mathcal{F}(L_t)$ without changing~$L_t$. An unblocked split causes a transition to a new linearization $L_{t+h} \neq L_t$ with $A(L_{t+h}) > A(L_t)$ by Theorem~\ref{thm:strict-improvement}.

    Consider any state $A \in \mathcal{F}(L_t)$. The deterministic variant of SFL (Algorithm~\ref{alg:sfl-deterministic}), which uses Bland's rule for all edge selections, terminates in a finite number of steps from any initial valid linearization by Theorem~\ref{thm:terminate}. In particular, starting from $A$, the deterministic algorithm produces a finite sequence of 
    improvement steps (each comprising a split selection and, for a blocked split, the subsequent remerge selection)
    that leads to 
    either an unblocked split or termination (when no valid split candidates remain).
    Let $N(A)$ denotes the number of steps the deterministic algorithm takes 
    from $A$ to reach this outcome.
    Since $\mathcal{F}(L_t)$ is finite, the quantity $N^* = \max_{A \in \mathcal{F}(L_t)} N(A)$ is also finite.

    The randomized algorithm draws its split from the first candidate-bearing tree in $Q$, whereas the deterministic algorithm selects the globally minimum-index candidate, possibly in a different tree. Since $Q$ is a uniformly random permutation of the active trees after every \textsc{Improve} (line~\ref{alg-line:rand-queue-reset}), the deterministic candidate's tree precedes all other candidate-bearing trees with probability at least $1/|V|$; this also holds mid-scan, as conditioned on the popped prefix (which contains no split candidates), the remainder of $Q$ is a uniformly random permutation of the remaining trees. Given the correct tree, the uniform draws from $\mathcal{E}_{\text{split}}$ (line~\ref{alg-line:rand-split-set}) and $\mathcal{E}_{\text{remerge}}$ (line~\ref{alg-line:rand-merge-rule3}) each match the deterministic choice with probability at least $1/|E|$. Hence, each step reproduces the deterministic step with probability at least $(|V|\cdot|E|^2)^{-1}$, and from any $A \in \mathcal{F}(L_t)$, the randomized algorithm follows the full deterministic path with probability at least
    \begin{align}
        \delta = \left(|V|\cdot|E|^{2}\right)^{-N^*} > 0. \label{eq:delta}
    \end{align}

    If the randomized algorithm deviates from this path at any step, it either performs an unblocked split (reaching a strictly better linearization, which is an even more favorable outcome) or arrives at a different state $A' \in \mathcal{F}(L_t)$, from which the same argument applies: the deterministic algorithm provides a path of length at most~$N^*$ to either an unblocked split or termination, and the randomized algorithm follows it with probability at least~$\delta$. 
    Here ``termination'' of the deterministic variant means it reaches a state with $\mathcal{E}_{\text{split}} = \emptyset$, i.e., no active tree admits an improving split. %
    Suppose this coincides with the randomized algorithm's halting condition of an empty queue $Q$ (we shall verify this shortly). In every consecutive window of $N^*$~steps spent on the plateau~$\mathcal{F}(L_t)$, the probability of either escaping the plateau or terminating is at least~$\delta > 0$. The probability of remaining on the plateau without terminating for $k$ such windows is at most $(1 - \delta)^k \to 0$ as $k \to \infty$. Therefore, with probability~$1$, 
    within finitely many steps, $h$, the algorithm either terminates or reaches a new linearization $L_{t+h}$ with $A(L_{t+h}) > A(L_t)$.

    It remains to confirm that the randomized algorithm's queue $Q$ eventually empties. After every successful \textsc{Improve}, $Q$ is reset to all active trees (line~\ref{alg-line:rand-queue-reset}), and a tree is popped without re-enqueueing (line~\ref{alg-line:rand-pop}) exactly when it has $\mathcal{E}_{\text{split}} = \emptyset$. Thus $Q$ becomes empty precisely when every active tree has no split candidate, which is the deterministic termination condition. As shown above, the algorithm almost surely reaches such a state in finitely many steps; there, every remaining tree is popped without re-enqueueing and $Q$ empties. Hence the randomized algorithm terminates with probability~$1$.

    Since the total number of distinct linearizations is finite and the AUC strictly increases with each plateau escape, the algorithm must 
    terminate with probability~$1$.
\end{proof}

\subsection{Granularity via Perturbation}
\label{ss:sfl-granularity}

As outlined in Sec.~\ref{ss:desirable-properties}, a robust linearization algorithm should provide \emph{granularity}: when multiple closures achieve the same maximum fee rate, the algorithm should systematically prefer the smallest among them. To enforce this property, we employ \emph{perturbation}. By adding an infinitesimal, transaction-specific offset to each fee, ordered lexicographically, we ensure that no two distinct subsets of transactions possess the exact same effective fee rate. This approach is inspired by Charnes' lexicographic perturbation method for resolving degeneracy in the simplex algorithm, where infinitesimal objective perturbations are used to impose a strict ordering among otherwise equivalent optima~\cite{charnes1952optimality}. As a result, the selection of the ``best'' closure becomes unique and well defined. Granularity is achieved in a separate minimization phase that runs after SFL.

We assign a unique, fixed integer index $\iota(i) \in \{1, \dots, n\}$ to every transaction $i \in V$. To promote fairness and robustness, these indices are assigned uniformly at random, so that in case of
early termination during the minimization phase, %
all chunks receive equal effort.
We conceptually replace the scalar fee $f_i$ with a polynomial in a symbolic variable $\epsilon$, where $\epsilon$ is strictly positive but infinitesimally small. We define two perturbed fee functions:
\begin{align}
    f_i^+(\epsilon) &= f_i + \epsilon^{\iota(i)}, \\
    f_i^-(\epsilon) &= f_i - \epsilon^{\iota(i)}.
\end{align}
The size $s_i$ remains unchanged. The \emph{perturbed fee rates} of a subset $U$ for the positive and negative cases are given %
by:
\begin{align}
    r_\epsilon^\pm(U) &= \frac{\sum_{i \in U} (f_i \pm \epsilon^{\iota(i)})}{\sum_{i \in U} s_i} = r(U) \pm 
    \frac1{s(U)} \sum_{i \in U} \epsilon^{\iota(i)} .
\end{align}

Comparing two perturbed fee rates is equivalent to a \emph{lexicographical comparison}. We define the strict positive ordering $A >_{\text{lex}}^+ B$ and the strict negative ordering $A >_{\text{lex}}^- B$ as follows:
\begin{enumerate}
    \item If the scalar fee rates differ ($r(A) \ne r(B)$), then both $A >_{\text{lex}}^+ B$ and $A >_{\text{lex}}^- B$ hold if and only if $A\succ B$. %
    \item If $r(A) = r(B)$, we compare the coefficients of the polynomials in $\epsilon$. Let $c_k(U) = \frac{\mathds{1}_{\{k \in U\}}}{s(U)}$ be the absolute coefficient of the $\epsilon^k$ term for a subset $U$. We identify the smallest transaction index $k$ such that $c_k(A) \ne c_k(B)$.
    \begin{itemize}
        \item %
        $A >_{\text{lex}}^+ B$ if and only if $c_k(A) > c_k(B)$.
        \item %
        $A >_{\text{lex}}^- B$ if and only if $c_k(A) < c_k(B)$.
    \end{itemize}
    When comparing disjoint sets $A$ and $B$, this rule simplifies to identifying the transaction $k^* = \min \{i \mid i \in A \cup B\}$.
    For the positive (resp.\ negative) perturbation, the set containing $k^*$ is strictly greater (resp.\ smaller).
\end{enumerate}

\begin{lemma}[Uniqueness]
\label{lem:uniqueness}
    For any two distinct subsets $A, B \subseteq V$, it holds that $r_\epsilon^+(A) \ne r_\epsilon^+(B)$ and $r_\epsilon^-(A) \ne r_\epsilon^-(B)$.
\end{lemma}
\begin{proof}
    Suppose $r_\epsilon^+(A) = r_\epsilon^+(B)$ or $r_\epsilon^-(A) = r_\epsilon^-(B)$. In either case, this implies that for every power of $\epsilon$ (i.e., for every transaction index $k \in \{1, \dots, n\}$), the absolute coefficients must be identical:
    \begin{align}
        \frac{\mathds{1}_{\{k \in A\}}}{s(A)} = \frac{\mathds{1}_{\{k \in B\}}}{s(B)}.
    \end{align}
    Summing these terms over all $k$ yields $\frac{|A|}{s(A)} = \frac{|B|}{s(B)}$. More importantly, since the denominators are constants, the indicator $\mathds{1}_{\{k \in A\}}$ can only equal $\mathds{1}_{\{k \in B\}}$ for all $k$ if the sets of indices are identical. Thus, $A$ and $B$ must be the same set.
\end{proof}

Since no two distinct subsets share the same perturbed fee rate (Lemma~\ref{lem:uniqueness}), we may assume strict inequality in all perturbed fee rate comparisons between distinct subsets.

\begin{algorithm}
\caption{Minimization}
\label{alg:minimization}
\begin{algorithmic}[1]
\STATE \textbf{Procedure:} \textsc{Minimize}($A$, $(V,E,f,s)$):
\STATE $Q \gets$ active trees of $(V,A)$ \label{alg-line:min-init}
\WHILE{$Q \neq \emptyset$}
    \STATE $T \gets Q.\text{pop}()$
    \STATE $U \gets$ vertex set of $T$ \label{alg-line:min-pop}
    \STATE $E_U \gets \{(u,v) \in E \mid u,v \in U\}$
    \STATE $s_U \gets (s_i)_{i \in U}$ \label{alg-line:min-extract}
    \STATE Assign unique indices $\iota(i) \in \{1,\dots,|U|\}$ for $\forall i \in U$ \label{alg-line:min-indices}
    \FOR{$\delta \in \{+, -\}$} \label{alg-line:min-for-sign}
        \STATE $f_U^\delta \gets (f_i^\delta(\epsilon))_{i \in U}$ \label{alg-line:min-perturb}
        \STATE $(\cdot,\,\cdot,\,A_U) \gets \textsc{RandomizedSFL}((U,E_U,f_U^\delta,s_U))$ \label{alg-line:min-call-sfl}
        \IF{$(U,A_U)$ has more than one active tree}
            \STATE $A \gets (A \setminus E_U) \cup A_U$ \label{alg-line:min-update-A}
            \STATE Push active trees of $(U,A_U)$ onto $Q$;\quad \textbf{break} \label{alg-line:min-push}
        \ENDIF
    \ENDFOR
\ENDWHILE
\STATE $(T_1, \dots, T_m) \gets$ active trees of $(V,A)$, sorted by descending fee rate, ties broken %
by ascending size $s(\cdot)$, subject to parent trees preceding their child trees \label{alg-line:min-sort}
\STATE $\sigma_{\text{out}} \gets$ a valid ordering induced by $(T_1, \dots, T_m)$
\STATE \textbf{Return:} $(T_1, \dots, T_m),\, \sigma_{\text{out}},\, A$
\end{algorithmic}
\end{algorithm}

To enforce granularity, we introduce a minimization phase (Algorithm~\ref{alg:minimization}) following the execution of the randomized SFL algorithm (with unperturbed fees).
When evaluating a candidate split that partitions a chunk into a parent chunk and a child chunk, the transaction possessing the absolute lowest index, which exerts the strictly dominant lexicographical influence, must reside in exactly one of these two disjoint sets. It cannot be known \textit{a priori} whether this dominating transaction falls into the parent or the child side of the boundary. Therefore, we independently evaluate the graph under both a positive perturbation, which successfully exposes the split if the lowest index is in the parent side, and a negative perturbation, which exposes the split if the lowest index is in the child side. By alternating these two perturbation passes, we guarantee that any valid equal-fee-rate boundary will be exposed.

\begin{theorem}[minimality of linearization]
\label{thm:minimality}
    Given $(V,A)$ corresponding to an optimal linearization, Algorithm~\ref{alg:minimization} produces an updated $(V,A)$ whose active trees form a minimal optimal linearization. That is, for any active tree $T$ of the output $(V,A)$, there exists no valid split candidate $(p,c)$ that partitions $T$ into $H_p$ and $H_c$ without triggering a self-merge such that $r(H_p) = r(H_c) = r(T)$.
\end{theorem}

\begin{proof}
    Suppose, for contradiction, that 
    an active tree $T$ of the output $(V,A)$ 
    is not minimal. Thus, there exists a
    valid split candidate $(p,c)$ that partitions $T$ into $H_p$ and $H_c$ 
    such that $r(H_p) = r(H_c) = r(T)$. Because 
    $T$ remains an active tree in the output $(V,A)$,
    it must have survived both perturbation passes intact.
    Let $i^* = \arg\min_{i \in T} \iota(i)$ be the transaction with the lowest index in $T$.
    It must reside in exactly one of $H_p$ or $H_c$.
    
    Case 1 ($i^* \in H_p$): Under the positive perturbation $f^+(\epsilon)$, $H_p$ receives the strictly dominant positive coefficient. Thus, $r_\epsilon^+(H_p) > r_\epsilon^+(H_c)$. SFL must split $T$ during 
    the $\delta = +$ pass,
    yielding a contradiction.

    Case 2 ($i^* \in H_c$): Under the negative perturbation $f^-(\epsilon)$, $H_c$ receives the strictly dominant negative index. Thus, $r_\epsilon^-(H_p) > r_\epsilon^-(H_c)$. SFL must split $T$ during 
    the $\delta = -$ pass,
    yielding a contradiction.

    Hence, $T$ is guaranteed to split in at least one pass. Therefore, no non-minimal 
    active tree can remain in the output of the algorithm.
\end{proof}

Since the original randomized SFL algorithm terminates with probability 1, and each successful split in Algorithm~\ref{alg:minimization} strictly decreases the finite size of the evaluated chunks, this minimization process is also guaranteed to terminate with probability 1.

\section{Comparing the Approaches}
\label{s:compare-approaches}

To evaluate the performance of different linearization algorithms, we conducted a series of benchmarks across diverse datasets. Each dataset was evaluated under four distinct linearization 
settings.\footnote{See 
\url{https://github.com/sipa/bitcoin_mempool_linearization_artefacts/} %
for the algorithm implementations, benchmark code, datasets, and results.} We compared the randomized SFL algorithm\footnote{The SFL implementation used in our experiments is very close to the production code deployed in Bitcoin Core~31.0. It is limited to 64 transactions, and exploits the ability to represent subsets of such clusters in a single machine integer.}
when run from scratch and when initialized with an optimal linearization (whose optimality is unknown to SFL) against the GGT approach\footnote{The specific algorithm is the paragraph~3.3 ``Finding all breakpoints of $\kappa(\lambda)$'' algorithm from the Gallo, Grigoriadis, and Tarjan paper~\cite{gallo1989fast}, applied to the flow network defined in the same paper in paragraph~4.2 ``Fractional programming applications'', under ``Maximum-ratio closure problem''. See Appendix~\ref{s:approach-ggt} for more information. Flows are represented as integers after multiplying with $2(\sum{} s_i)^2$, which is sufficient to guarantee distinct flows are distinct. The minimal optimal chunks can be found as the strongly-connected components of the residual flow graph within each found chunk, after running the push-relabel algorithm in them to completion, to find maximal flows rather than just pre-flows. Like the SFL implementation, several optimizations exploit the fact that transaction sets in clusters up to 64 transactions can be represented by a single machine integer.}
in both its full bidirectional form and a variant that processes each maximum-ratio closure subproblem in a uniformly random direction (forward or backward). 

The benchmarking methodology was as follows: For a list of $N$ clusters within a dataset, each cluster underwent the following process 
1024 times: 100 random number generator seeds were generated for the linearization algorithm. For each seed, the optimal linearization was benchmarked 
7 times, and the median of these 
7 runs was recorded to mitigate variance from system-level task swapping. The average of these 100 medians was then calculated, aiming to represent the typical duration for linearizing many clusters. 
Finally, for each cluster size, three statistics were reported over all clusters of that size and all 
1024 repetitions, with all times converted to microseconds ($\mu s$): the average of the average-of-medians (``avg''), representing typical performance; their maximum (``max-of-avg''), representing aggregate worst-case behavior over many clusters, which is the most relevant measure in practice;
and the maximum over all individual medians (``max''), representing the worst-case behavior of a single instance. All benchmarks were written in C++20, compiled with GCC~15.2, and executed on a 64-core AMD Ryzen Threadripper 9980X system running Ubuntu~26.04, 62-way parallel; each core processed a random subset of the cases in random order, so that cross-core interference largely cancels out in the reported medians.

\captionsetup{margin={9mm,1mm}} 
\begin{figure*}[htbp]
    \begin{minipage}{0.497\textwidth}
        \centering
        \includegraphics[width=\textwidth]{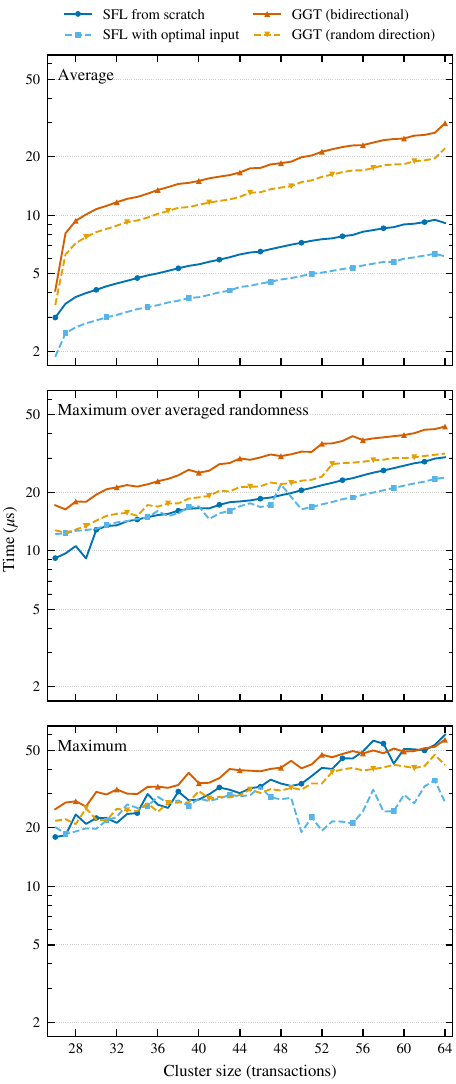}
        \caption{Runtime of SFL and GGT on the replayed 2023 mempool dataset.}
        \label{fig:sim2023}
    \end{minipage}
    \hfill %
    \begin{minipage}{0.497\textwidth}
        \centering
        \includegraphics[width=\textwidth]{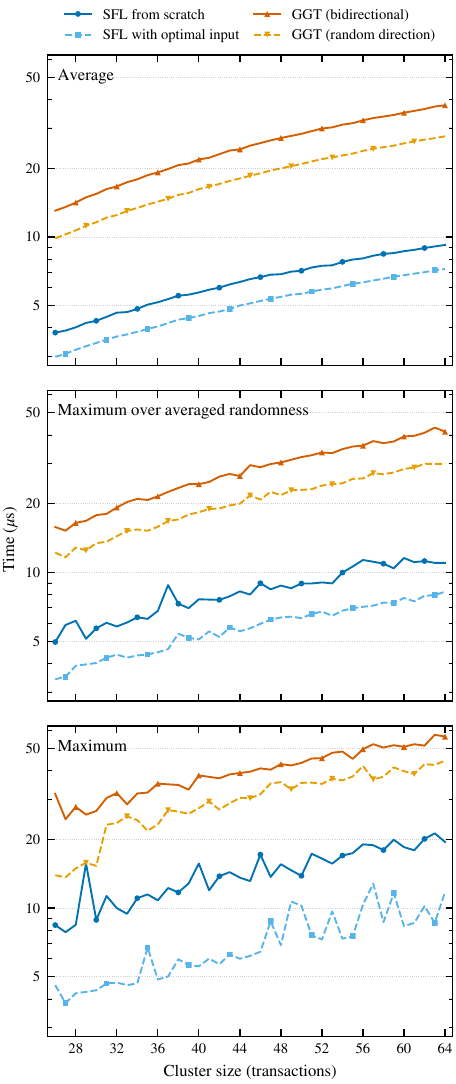}
        \caption{Runtime of SFL and GGT on a dataset of randomly generated tree-shaped clusters.}
        \label{fig:spanning}
    \end{minipage}
\end{figure*}

\begin{figure*}[htbp]
    \centering
    \begin{minipage}{0.497\textwidth}
        \centering
        \includegraphics[width=\textwidth]{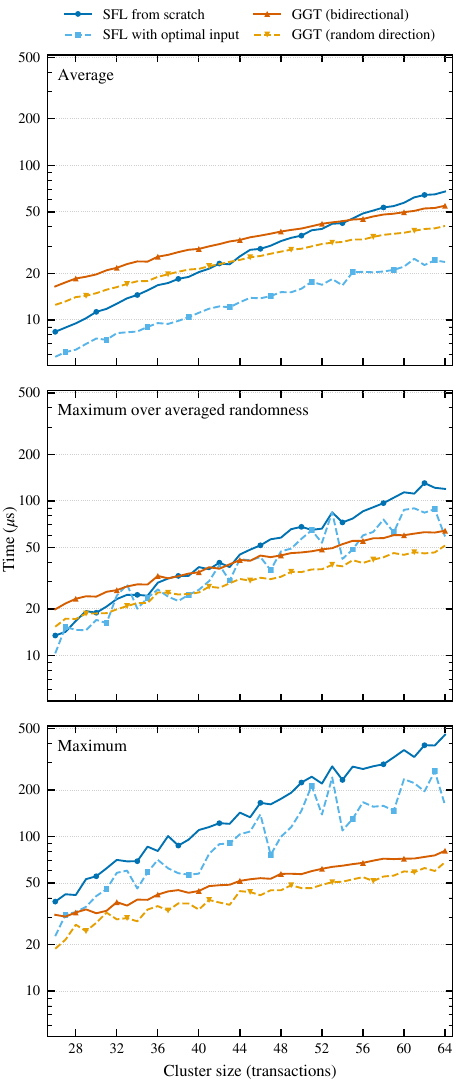}
        \caption{Runtime of SFL and GGT on a dataset of randomly-generated medium-density clusters.}
        \label{fig:medium}
    \end{minipage}
    \hfill
    \begin{minipage}{0.497\textwidth}
        \centering
        \includegraphics[width=\textwidth]{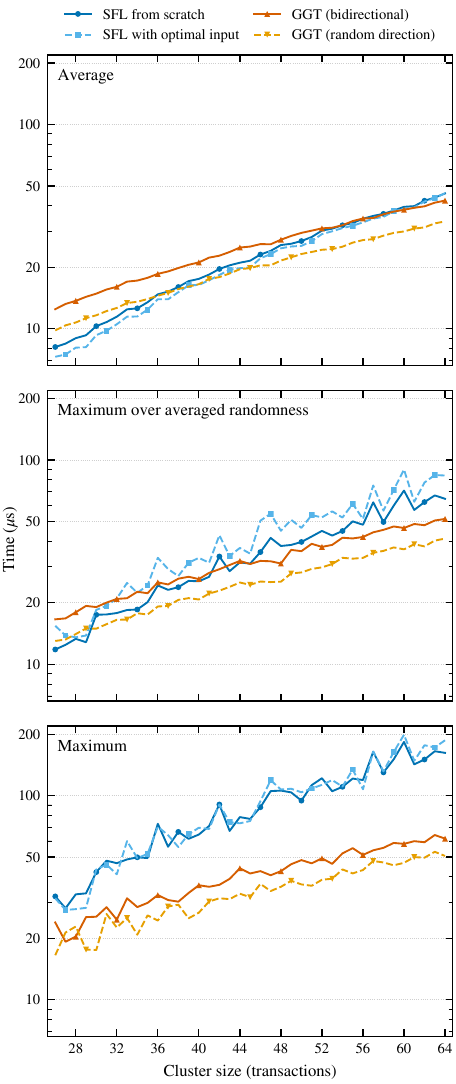}
        \caption{Runtime of SFL and GGT on a dataset of randomly-generated complete bipartite clusters.}
        \label{fig:bipartite}
    \end{minipage}
\end{figure*}
\captionsetup{margin={0mm,0mm}} 

The benchmarks were performed on several types of cluster data:

\begin{itemize}
    \item \textbf{Simulated 2023 Mempools:} These clusters were generated by replaying a dump of Bitcoin's 2023 P2P network activity into a cluster-mempool-patched version of Bitcoin Core.\footnote{This is a replay using a modern version of the codebase with cluster-based limits (up to 64 transactions), while the data being replayed is from a time when no cluster limits were enforced on the network (instead there was a limit of 24 ancestors and 24 descendants per transaction). This means that this replay is not necessarily representative of what transaction traffic would have looked like if cluster limits had been in place at the time.} This dataset included 199,076 clusters with sizes ranging from 26 to 64 transactions.\footnote{Bitcoin's previous policy (of max 24 ancestors/descendants per transaction) was introduced in September 2015 (\url{https://github.com/bitcoin/bitcoin/pull/6654}). The mempool cluster size limit is likely to remain within 64 in the foreseeable future. %
    }
    Fig.~\ref{fig:sim2023} presents the results for this dataset. SFL from scratch is roughly two times faster than both GGT variants in both average and seed-averaged worst-case runtime, and starting from an optimal linearization roughly halves its runtime again. On the worst individual instances (Fig.~\ref{fig:sim2023} (bottom)), %
    the from-scratch algorithms are largely comparable, while SFL with an optimal input retains a clear advantage.
\footnote{Due to the ancestors and descendants limit of 24 transactions at the time this data was captured, clusters of up to 49 likely consist primarily of two largely independent groups of up to 25 transactions, with a shared transaction connecting them. This may explain the sudden decline at size 49, where the captured clusters transition to three independent transaction groups.}    %

    \item \textbf{Randomly-Generated Tree-Shaped Clusters:} For each cluster size $n$ from 26 to 64 transactions, we generated 100 random tree-shaped clusters. These clusters represent connected graphs with exactly $n-1$ dependencies. 
    Fig.~\ref{fig:spanning} presents the results for this dataset. SFL is the clear winner across all three metrics %
    by a factor of roughly $2$--$3$ times over both GGT variants.

    \item \textbf{Randomly-Generated Medium-Density Clusters:} For each cluster size $n$ between 26 and 64 transactions, 100 random clusters with approximately 3 dependencies per transaction were generated. Fig.~\ref{fig:medium} presents the results for this dataset. Here, SFL's runtime inherently scales with the number of dependencies: while SFL from scratch remains fastest on average for smaller clusters, it is overtaken by both GGT variants as the cluster size approaches 64, and its individual worst case (Fig.~\ref{fig:medium} (bottom)) %
    is several times higher than GGT's. Starting SFL from an optimal linearization, however, keeps its average runtime well below all alternatives. Note also that randomizing GGT's direction visibly improves its performance across all metrics.

    \item \textbf{Randomly-Generated Complete Bipartite Clusters:} For each cluster size $n$ from 26 to 64 transactions, 100 random complete-connected bipartite clusters were generated, featuring approximately $n/4$ dependencies per transaction. These clusters are notable for their high number of dependencies. For instance, a 64-transaction cluster where all inputs and outputs are key-path Pay-to-Taproot (P2TR) would require approximately 103 kvB in transaction input and output data. Fig.~\ref{fig:bipartite} presents the results for this dataset. With even more dependencies per transaction, SFL's dependency-driven scaling is most visible: average runtimes of all algorithms remain within a narrow band, but in both worst-case metrics, GGT clearly outperforms SFL, whose two variants become nearly indistinguishable, as an optimal starting linearization no longer offsets the per-dependency cost.
\end{itemize}

These results demonstrate that the SFL algorithm is highly efficient in practice. On realistic mempool data and sparse synthetic clusters, it consistently outperforms the GGT-based approaches in average runtime, aggregate worst-case runtime, and frequently on individual worst-case instances. While GGT exhibits superior worst-case performance on dense synthetic clusters, pure worst-case performance %
is not the deciding criterion for deployment. SFL's ability to %
leverage an existing linearization---which the benchmarks show to be visibly beneficial---alongside %
its fairness and anytime properties discussed in Sec.~\ref{ss:desirable-properties}, is more relevant in a practical %
setting. %
In such environments,
per-cluster computation must be bounded, and no participant's transactions should be starved. Moreover, 
as the algorithm's internal randomness is insulated from external exploitation, the seed-averaged worst case 
(middle plots in Figs.~\ref{fig:sim2023}--\ref{fig:bipartite}) serves as the practically relevant %
measure. %
Under this metric, SFL remains %
competitive with GGT even on the densest clusters while being on average often faster. Combined, these properties make SFL well-suited for real-world deployment.

Note however that none of the benchmarked datasets are actual worst cases for any of the tested algorithms, as we do not know what they are. In the case of GGT we can however extrapolate from the benchmarked time per iteration to the known upper bound on iteration count. This predicts a time up to 10~ms for 64-transaction clusters, far too slow for our application. If this bound were tight, it would be a further justification for relying on randomization and fairness, rather than aiming to always achieve optimality.

\section{Conclusion}
\label{s:conclusion}

This paper has introduced the mempool linearization problem: determining an optimal partition of the set of transactions and their ordering that maximizes fee efficiency while respecting parent-child dependencies. We have established several theoretical characterizations of the problem. In particular, we showed that mempool linearization admits an LP formulation from which an exact optimal solution can be recovered, and that it is equivalent to a sequence of maximum-ratio closure problems solvable via the classical GGT parametric preflow algorithm. Our main algorithmic contribution is the SFL algorithm, a %
method that iteratively refines a partition of transactions through merge and split operations guided by fee-rate comparisons.

Extensive evaluation on both synthetic and real Bitcoin mempool data demonstrates that SFL consistently produces optimal linearizations while remaining competitive with the GGT-based approach in runtime.
Moreover, SFL naturally supports randomized execution, allowing computational effort to be distributed across the transaction graph and providing robustness under time constraints or early termination.

Beyond its algorithmic value, efficient mempool linearization has broader implications for blockchain decentralization. By enabling all participants to compute high-quality transaction orderings with modest computational resources, it reduces the advantage of proprietary infrastructure and private relay networks, thereby supporting a more open and censorship-resistant mining ecosystem.

Several directions remain for future work. On the theoretical side, it would be valuable to further analyze the complexity of SFL and develop stronger performance guarantees for its randomized variants. On the practical side, extending the framework to handle adversarially constructed transaction graphs and conflicting transactions would provide a more realistic model of mempool dynamics. More broadly, our results suggest that mempool linearization is not merely an implementation detail, but a core optimization problem at the intersection of economics, networking, and consensus. Developing efficient and robust algorithms for this problem will remain important for the continued scalability and decentralization of blockchain systems.

\section*{Acknowledgement}

The authors would like to thank Suhas Daftuar, Ittay Eyal, Gregory Maxwell, Stefan Richter, Clara Shikhelman, Anthony Towns, Aviv Yaish, and Aviv Zohar for stimulating discussions, technical suggestions, and feedback.

This work was supported in part by the National Science Foundation under Grant Nos.~2132700, 2216970, and 2434044.

\bibliographystyle{IEEEtran}
\bibliography{refs}

@article{gallo1989fast,
  title={A fast parametric maximum flow algorithm and applications},
  author={Gallo, Giorgio and Grigoriadis, Michael D and Tarjan, Robert E},
  journal={SIAM Journal on Computing},
  volume={18},
  number={1},
  pages={30--55},
  year={1989},
  publisher={SIAM}
}

@book{Ford1962FlowsIN,
  title={Flows in networks},
  author={Lester Randolph Ford and Delbert Ray Fulkerson},
  year={1962},
  publisher = {Princeton University Press},
  url={https://api.semanticscholar.org/CorpusID:260503522}
}

@article{Picard01111982,
author = {Picard, Jean-Claude and Queyranne, Maurice},
title = {Selected Applications of Minimum Cuts in Networks},
journal = {INFOR: Information Systems and Operational Research},
volume = {20},
number = {4},
pages = {394--422},
year = {1982},
publisher = {Taylor \& Francis},
doi = {10.1080/03155986.1982.11731876},
URL = { 
    https://doi.org/10.1080/03155986.1982.11731876
},
eprint = { 
    https://doi.org/10.1080/03155986.1982.11731876
}
}

@article{isbell1956attrition,
  author = {Isbell, J. R. and Marlow, W. H.},
  title = {Attrition Games},
  journal = {Naval Research Logistics Quarterly},
  volume = {3},
  number = {1-2},
  pages = {71-94},
  year = {1956},
  doi = {10.1002/nav.3800030108},
  url = {https://onlinelibrary.wiley.com/doi/abs/10.1002/nav.3800030108},
  eprint = {https://onlinelibrary.wiley.com/doi/pdf/10.1002/nav.3800030108}
}

@inproceedings{cheriyan1988preflow,
author = {Cheriyan, Joseph and Maheshwari, S. N.},
title = {Analysis of Preflow Push Algorithms for Maximum Network Flow},
year = {1988},
isbn = {3540505172},
publisher = {Springer-Verlag},
address = {Berlin, Heidelberg},
booktitle = {Proceedings of the Eighth Conference on Foundations of Software Technology and Theoretical Computer Science},
pages = {30–48},
numpages = {19}
}

@article{GT88Preflow,
author = {Goldberg, Andrew V. and Tarjan, Robert E.},
title = {A new approach to the maximum-flow problem},
year = {1988},
issue_date = {Oct. 1988},
publisher = {Association for Computing Machinery},
address = {New York, NY, USA},
volume = {35},
number = {4},
issn = {0004-5411},
url = {https://doi.org/10.1145/48014.61051},
doi = {10.1145/48014.61051},
journal = {J. ACM},
month = oct,
pages = {921–940},
numpages = {20}
}

@article{pardalos1991fractional,
  author = {Pardalos, P. M. and Phillips, A. T.},
  title = {Global Optimization of Fractional Programs},
  journal = {Journal of Global Optimization},
  volume = {1},
  number = {2},
  pages = {173--182},
  year = {1991},
  doi = {10.1007/BF00119990},
  url = {https://doi.org/10.1007/BF00119990},
  issn = {1573-2916}
}

@inproceedings{beinesmincut2025,
author = {Arne Beines and Michael Kaibel and Philip Mayer and Petra Mutzel and Jonas Sauer},
title = {A Simpler Approach for Monotone Parametric Minimum Cut: Finding the Breakpoints in Order},
booktitle = {2025 Proceedings of the Symposium on Algorithm Engineering and Experiments (ALENEX)},
chapter = {},
pages = {29-41},
year=2025,
doi = {10.1137/1.9781611978339.3},
URL = {https://epubs.siam.org/doi/abs/10.1137/1.9781611978339.3},
eprint = {https://epubs.siam.org/doi/pdf/10.1137/1.9781611978339.3}
}

@book{nocedal2006numerical,
  author = {Nocedal, Jorge and Wright, Stephen J},
  publisher = {Springer},
  title = {Numerical optimization},
  year = {2006}
}

@incollection{klee1972simplex,
  author    = {Victor Klee and George J. Minty},
  title     = {How Good is the Simplex Algorithm?},
  booktitle = {Inequalities},
  editor    = {Olga Shisha},
  publisher = {Academic Press},
  address   = {New York},
  year      = {1972},
  pages     = {159--175}
}

@article{spielman2004smoothed,
  title={Smoothed analysis of algorithms: Why the simplex algorithm usually takes polynomial time},
  author={Spielman, Daniel A and Teng, Shang-Hua},
  journal={Journal of the ACM (JACM)},
  volume={51},
  number={3},
  pages={385--463},
  year={2004},
  publisher={ACM New York, NY, USA}
}

@book{diestel2017graph,
  title={Graph Theory},
  author={Diestel, Reinhard},
  year={2017},
  edition={5},
  publisher={Springer},
  series={Graduate Texts in Mathematics},
  volume={173}
}

@inproceedings{kalai1992subexponential,
  title={A subexponential randomized simplex algorithm},
  author={Kalai, Gil},
  booktitle={Proceedings of the twenty-fourth annual ACM symposium on Theory of computing},
  pages={475--482},
  year={1992}
}

@inproceedings{matouvsek1992subexponential,
  title={A subexponential bound for linear programming},
  author={Matou{\v{s}}ek, Ji{\v{r}}{\'\i} and Sharir, Micha and Welzl, Emo},
  booktitle={Proceedings of the eighth annual symposium on Computational geometry},
  pages={1--8},
  year={1992}
}

@article{hochbaum2001new,
  title={A new—old algorithm for minimum-cut and maximum-flow in closure graphs},
  author={Hochbaum, Dorit S},
  journal={Networks: An International Journal},
  volume={37},
  number={4},
  pages={171--193},
  year={2001},
  publisher={Wiley Online Library}
}

@article{lerchs1965optimum,
  title={Optimum design of open-pit mines},
  author={Lerchs, H. and Grossmann, I. F.},
  journal={Transactions of the Canadian Institute of Mining and Metallurgy},
  volume={68},
  pages={17--24},
  year={1965}
}

@article{charnes1952optimality,
  title={Optimality and Degeneracy in Linear Programming},
  author={Charnes, Abraham},
  journal={Econometrica},
  volume={20},
  number={2},
  pages={160--170},
  year={1952},
  publisher={Wiley}
}

@article{bland1977new,
  title={New finite pivoting rules for the simplex method},
  author={Bland, Robert G},
  journal={Mathematics of Operations Research},
  volume={2},
  number={2},
  pages={103--107},
  year={1977},
  publisher={INFORMS}
}

@article{pahari2025exclusive,
  title={How Exclusive are Ethereum Transactions? Evidence from non-winning blocks},
  author={Pahari, Vabuk and Canidio, Andrea},
  journal={arXiv preprint arXiv:2509.16052},
  year={2025}
}

@article{sidney1975decomposition,
  title={Decomposition algorithms for single-machine sequencing with precedence relations and deferral costs},
  author={Sidney, Jeffrey B},
  journal={Operations Research},
  volume={23},
  number={2},
  pages={283--298},
  year={1975},
  publisher={INFORMS}
}

@incollection{lawler1978sequencing,
  title={Sequencing jobs to minimize total weighted completion time subject to precedence constraints},
  author={Lawler, Eugene L},
  booktitle={Annals of discrete mathematics},
  volume={2},
  pages={75--90},
  year={1978},
  publisher={Elsevier}
}

@incollection{potts2009algorithm,
  title={An algorithm for the single machine sequencing problem with precedence constraints},
  author={Potts, Chris N},
  booktitle={Combinatorial optimization ii},
  pages={78--87},
  year={2009},
  publisher={Springer}
}

@article{chudak1999half,
  title={A half-integral linear programming relaxation for scheduling precedence-constrained jobs on a single machine},
  author={Chudak, Fabi{\'a}n A and Hochbaum, Dorit S},
  journal={Operations Research Letters},
  volume={25},
  number={5},
  pages={199--204},
  year={1999},
  publisher={Elsevier}
}

@article{Smith1956VariousOF,
  title={Various optimizers for single‐stage production},
  author={Wayne E. Smith},
  journal={Naval Research Logistics Quarterly},
  year={1956},
  volume={3},
  pages={59-66},
  url={https://api.semanticscholar.org/CorpusID:120614124}
}

@book{bertsimas1997introduction,
  title={Introduction to linear optimization},
  author={Bertsimas, Dimitris and Tsitsiklis, John N},
  volume={6},
  year={1997},
  publisher={Athena scientific Belmont, MA}
}

@article{charnes1962programming,
  title={Programming with linear fractional functionals},
  author={Charnes, Abraham and Cooper, William W},
  journal={Naval Research logistics quarterly},
  volume={9},
  number={3-4},
  pages={181--186},
  year={1962},
  publisher={Wiley Online Library}
}

\appendices

\section{Proof
Theorem~\ref{thm:auc-optimality}}
\label{a:auc-optimality}

    Given any closure $W$ in $G=(V,E)$,
    we 
    let 
    $W_j = W \cap U^*_j$ for $j=1, \dots, p$. We claim
    \begin{align} \label{eq:fwj}
        f(W_j) \le s(W_j) r( U^*_j ) %
    \end{align}    
    holds for all $j=1,\dots,p$. We prove~\eqref{eq:fwj} by induction.
    
    We begin with $j=1$. If $W_1=\emptyset$,~\eqref{eq:fwj} holds trivially for $j=1$. As $W$ and $U^*_1$ are closures in $G$, their intersection $W_1$, if non-empty, is a closure in $G$ by Lemma~\ref{lm:closures_properties}. %
    Because $U^*_1$ is %
    a highest-fee-rate closure, %
    $r(W_1) \le r(U^*_1)$, which implies~\eqref{eq:fwj} for $j=1$.
    Next, assume that for some $m\in\{1,\dots,p\}$, %
    $r(W_j) \le r(U^*_j)$ for all $j=1,\dots,m$. Let $W'_{m} = W \setminus \bigcup_{j=1}^m W_j = W \setminus \bigcup_{j=1}^m U^*_j$. By Lemma~\ref{lem:prefix-closure}, $\bigcup_{j=1}^m U^*_j$ is a closure in $G$. %
    Hence, by 
    Lemma~\ref{lem:closure-difference}, $W'_{m}$, if non-empty, is a closure in the subgraph induced by $V \setminus \bigcup_{j=1}^m U^*_j$.
    We have $W_{m+1} = W \cap U^*_{m+1} = W'_{m} \cap U^*_{m+1}$. Since $U^*_{m+1}$ is a closure in the subgraph induced by $V \setminus \bigcup_{j=1}^m U^*_j$, the intersection $W_{m+1}$, if non-empty, is also a closure in that same subgraph (Lemma~\ref{lm:closures_properties}).
    Since $U^*_{m+1}$ is the optimal closure in that subgraph, it must hold that $r(W_{m+1}) \le r(U^*_{m+1})$, i.e.,~\eqref{eq:fwj} holds for $j=m+1$. By induction,~\eqref{eq:fwj} holds for all $j=1,\dots,p$.
    Summing~\eqref{eq:fwj} over $j=1,\dots,p$ %
    yields
    \begin{align}
        f(W)
        &= \sum_{j=1}^p f(W_j) \\
        &\le \sum_{j=1}^p s(W_j) r(U^*_j). \label{eq:fw<=sum}
    \end{align}

    We next show that every breakpoint of $C_L$ lies on or below $C_{L^*}$. {The $k$-th breakpoint of $C_L$ takes the form of $(s(W),f(W))$ with}
    $W = \bigcup_{i=1}^k U_i$. By Lemma~\ref{lem:prefix-closure}, $W$ is a closure in $G$. %
    It suffices to show $f(W) \le C_{L^*}(s(W))$.

    The value $C_{L^*}(s(W))$ is the cumulative fee of the first $s(W)$ size units of $L^*$. Since $r(U^*_1) \ge r(U^*_2) \ge \dots$ by Definition~\ref{def:L},
    $C_{L^*}(s(W))$ is the result of the following maximization:
    \begin{subequations}\label{eq:knapsack}
    \begin{align}
        \maximize_{x_1,\dots,x_p} \quad & \sum_{j=1}^{p} x_j \, r(U^*_j) \label{eq:knapsack:a} \\
        \subjectto \quad 
        & \sum_{j=1}^{p} x_j = s(W), \label{eq:knapsack:b} \\
        & 0 \le x_j \le s(U^*_j), \quad j = 1, \dots, p . \label{eq:knapsack:c}
    \end{align}
    \end{subequations}
    Since $W_j \subseteq U^*_j$ implies $s(W_j) \le s(U^*_j)$, letting $x_j=s(W_j)$ for 
    $j=1,\dots,p$
    is a feasible solution. Therefore,
    \begin{align} \label{eq:sumsr<=}
        \sum_{j=1}^p s(W_j) r(U^*_j) \le C_{L^*}(s(W)) .
    \end{align}
    By~\eqref{eq:fw<=sum} and~\eqref{eq:sumsr<=}, we obtain
    $f(W) \le C_{L^*}(s(W))$ %
    for every $k=1,\dots,q$.

    Finally, since every breakpoint of $C_L$ is dominated by $C_{L^*}$, which is a piece-wise linear concave function, the entire function $C_L$ is dominated by $C_{L^*}$. Therefore, the AUC of $L^*$ also dominates the AUC of $L$.

\section{Proof Theorem~\ref{thm:lp=bp}}
\label{a:lp=bp}

Let $x_{1},...,x_{n}$ be a solution to LP~\ref{eq:lp}. For convenience, let $A$ consist of indices of the $x$-variables that are positive, i.e., $A=\{i \in V:x_{i}>0\}$. Let
\begin{align}
    u=\max\{x_{j} \mid j\in V\}    
\end{align}
denote the maximum value of the $x$-variables. Let $B=\{i\in V \mid x_{i}=u\}$, i.e., the set of indices for which the $x$-variable attains the maximum. Evidently, $z_i=\mathds{1}_{\{i\in B\}}$.

We first show that $x_{1}^{\prime},...,x_{n}^{\prime}$ defined as
\begin{align}
\label{eq:x'-def}
    x_{i}^{\prime}= %
    \frac1{\sum_{i\in B}s_i}
    \cdot\mathds{1}_{\{i\in B\}}, \quad i=1,...,n
\end{align}
also solves the LP in two separate cases: i) $A=B$ and ii) $A\ne B$. One can readily verify that $\sum_{i \in V} s_i x_i^{\prime} = 1$ and that $x_i^{\prime} \geq x_j^{\prime}$ for all $(i,j)\in E$. %
Thus, $x_1^{\prime}, \ldots, x_n^{\prime}$ satisfies all constraints of the LP. The claim that $x'_1,\dots,x'_n$ solves the LP is trivial %
in case (i), %
when all strictly positive entries in $x_1,\dots,x_n$ are identical.
In case (ii), $A\ne B$, %
we let $w=\max\{x_{j} \mid j\in A \backslash B\}$, which denotes the second highest value the $x$-variables attain, which must satisfy
$w>0$. Consider $x_{1}^{\prime\prime},...,x_{n}^{\prime\prime}$ defined by
\begin{align}
\label{eq:x''-def}
    x_{i}^{\prime\prime}=\begin{cases}
        w/q,& \text{ if } i\in B,\\ 
        x_i/q,& \text{ otherwise,}
    \end{cases}
\end{align}
where $q=(\sum_{i \in B}s_{i})w + \sum_{i \in V\backslash B} s_ix_i$. It is easy to verify that $\sum_{i \in V}s_{i}x_{i}^{\prime\prime}=1$ and $x_{i}^{\prime\prime} \ge x_{j}^{\prime\prime}$ for all $(i,j)\in E$.
Hence $x_{1}^{\prime\prime},...,x_{n}^{\prime\prime}$ satisfies all the constraints of the LP. 
The %
result of the LP can be expressed as
\begin{align}
\label{eq:lp-rate}
    r_{\text{LP}}^{*}
    &= \sum_{i \in V} f_i x_i \\
    &= \frac{\sum_{i \in V} f_i x_i}{\sum_{i \in V} s_i x_i} \\
    &= \frac{(u - w) (\sum_{i \in B} f_i) + q \sum_{i \in V} f_i x''_i}{(u - w) (\sum_{i \in B} s_i) + q \sum_{i \in V} s_i x''_i}. \label{eq:rLP*=}
\end{align}
We show that
\begin{align}
\label{eq:rate-equal}
\frac{\sum_{i \in V} f_i x''_i}{\sum_{i \in V} s_i x''_i} = \frac{\sum_{i \in B} f_i}{\sum_{i \in B} s_i}
\end{align}
by contradiction: If either side of~\eqref{eq:rate-equal} is strictly greater than the other side,   
then~\eqref{eq:rLP*=} implies that $r_{\text{LP}}^*$ is strictly in between those two sides. Since the right hand side of~\eqref{eq:rate-equal} is the fee rate achieved by $x_{i}^{\prime}$ defined by~\eqref{eq:x'-def} due to
\begin{align}
    \frac{\sum_{i \in B} f_i}{\sum_{i \in B} s_i} = \frac{\sum_{i \in V} f_i x'_i}{\sum_{i \in V} s_i x'_i} ,
\end{align}
and the left-hand side is the fee rate achieved by $x_{i}^{\prime\prime}$ defined by~\eqref{eq:x''-def}, $r_{\text{LP}}^*$ cannot be the maximum fee rate, contradicting its definition.
With~\eqref{eq:rate-equal} established, it is straightforward to see that $x_{i}^{\prime}$ achieves the highest fee rate $r_\text{LP}^*$, i.e., it also solves the same LP.

By~\eqref{eq:zi-def} and~\eqref{eq:x'-def}, we have $z_{i}=\mathds{1}_{\{x_{i}^{\prime}>0\}}$ for all $i \in V$. To establish the theorem, it suffices to show that $z_{1},...,z_{n}$ thus defined solves the BP in~\eqref{eq:bp}.
Since FP~\eqref{eq:fp} can be viewed as the BP in~\eqref{eq:bp} with the binary constraints relaxed, we have %
\begin{align}
\label{eq:r_fp>r_bp}    
    r_{\text{FP}}^{*}\ge r_{\text{BP}}^{*} .
\end{align}
Meanwhile, the FP is equivalent to the LP in~\eqref{eq:lp} via the Charnes-Cooper transformation~\cite{charnes1962programming}, in the sense that there is a one-to-one correspondence between their solutions in the form of
\begin{align}
\label{eq:fp-to-lp}
x_i = y_i \Big/ {\sum_{k \in V} s_k y_k}, \quad \forall i \in V
\end{align}
which implies that they achieve the same maximum:
\begin{align}
\label{eq:lp-fp-equal}
\sum_{i \in V} f_i x_i = \frac{\sum_{i \in V} f_i y_i}{\sum_{i \in V} s_i y_i} .
\end{align}
This, together with~\eqref{eq:r_fp>r_bp}, leads to
\begin{align}
\label{eq:r_lp=r_fp>=r_bp}
r_{\text{LP}}^{*} = r_{\text{FP}}^{*} \ge r_{\text{BP}}^{*} .
\end{align}

Note that if $y_{1},...,y_{n}$ solves the FP, scaling all $y_{i}$'s by the same positive constant still solves the FP. Since $x_{1}^{\prime},...,x_{n}^{\prime}$ solves the LP, $z_{1},...,z_{n}$ which are defined by~\eqref{eq:zi-def}, must solve the FP to achieve $r_{FP}^{*}$. Evidently, $z_{1},...,z_{n}$ also satisfies all the constraints of the BP in~\eqref{eq:bp}, which implies that $r_{BP}^{*}\ge r_{FP}^{*}$. This, together with \eqref{eq:r_lp=r_fp>=r_bp} implies that $r_{\text{BP}}^{*}=r_{\text{FP}}^{*}=r_{\text{LP}}^{*}$. Therefore, the same $z_{1},...,z_{n}$ must also solve the BP.

\section{Proof of~\eqref{eq:gammaq}}
\label{a:gammaq}

The left hand side of~\eqref{eq:gammaq} is equal to
\begin{align}
    &
    \sum_{i,j\in V} f_i s_j
    -
    \sum_{i,j\in V: \sigma(j)\le\sigma(i)} f_i s_j \notag \\
    &= \sum_{i,j\in V: i\ne j} f_i s_j
    -
    \sum_{i,j\in V: \sigma(j)<\sigma(i)} f_i s_j \\
    &= \frac12 \sum_{\substack{i,j\in V\\ i\ne j}} f_i s_j
    + \frac12 \sum_{\substack{i,j\in V\\ \sigma(j)>\sigma(i)}} f_i s_j 
    - \frac12 \sum_{\substack{i,j\in V\\ \sigma(j)<\sigma(i)}} f_i s_j \\
    &= \frac12 \sum_{i,j\in V: i\ne j} f_i s_j
    + \frac12 \sum_{i,j\in V: \sigma(j)>\sigma(i)} (f_i s_j - f_j s_i)
\end{align}
which is equal to the right hand side of~\eqref{eq:gammaq}.

\section{Refined SFL with Maximum-$q$ Heuristic}
\label{s:sfl-max-q}

This appendix describes the maximum-$q$ heuristic for split selection in the refined SFL algorithm, referenced in Sec.~\ref{ss:spl-refinements}. This deterministic rule selects, at each improvement step, the active dependency that maximizes the $q$-function value among all valid split candidates. 

\begin{algorithm}
\caption{Refined SFL with Maximum-$q$ Heuristic}
\label{alg:sfl-refined-max-q}
\begin{algorithmic}[1]
\STATE \textbf{Procedure:} \textsc{Improve}($(p,c)$, $A$):
\STATE $\mathcal{E}_{\text{remerge}} \gets \{(u,v) \in E \setminus A \mid u \in H_c,\; v \in H_p\}$
\STATE $A \gets A \setminus \{(p,c)\}$
\IF{$\mathcal{E}_{\text{remerge}} \neq \emptyset$}
    \STATE Pick arbitrary $e' \in \mathcal{E}_{\text{remerge}}$;\quad $A \gets A \cup \{e'\}$
\ELSE
    \STATE $A \gets \textsc{MergeUpwards}(p, A)$
    \STATE $A \gets \textsc{MergeDownwards}(c, A)$
\ENDIF
\RETURN $A$

\vspace{0.5em}
\STATE \textbf{Procedure:} \textsc{MaxQSFL}(Mempool $(V,E,f,s)$, optional valid ordering $\sigma$):
\STATE $(\cdot,\, \cdot,\, A) \gets \textsc{Linearize}((V,E,f,s),\, \sigma)$
\WHILE{true}
    \STATE $\mathcal{E}_{\text{split}} \gets \{(p,c) \in A \mid H_p \succ H_c\}$
    \IF{$\mathcal{E}_{\text{split}} = \emptyset$}
        \STATE \textbf{break}
    \ENDIF
\STATE $e^* \gets \arg\max_{e=(p,c) \in \mathcal{E}_{\text{split}}} q(H_p,H_c)$, with ties broken arbitrarily    \STATE $A \gets \textsc{Improve}(e^*, A)$
\ENDWHILE
\STATE $(T_1,\dots,T_m) \gets$ active trees of $(V,A)$, sorted by descending fee rate, ties broken arbitrarily
\STATE $\sigma_{\text{out}} \gets$ a valid ordering induced by $(T_1,\dots,T_m)$
\STATE \textbf{Return:} $(T_1,\dots,T_m),\, \sigma_{\text{out}}, A$
\end{algorithmic}
\end{algorithm}

Returning to the LP interpretation, a split corresponds to making a slack variable 
basic.
The simplex method only permits such a pivot if the derivative of the objective 
with respect to 
the slack variable
is positive, which, in our case, translates to requiring the parent-side chunk \( A \) to have a higher fee rate than the child-side chunk \( B \). A natural strategy is to select the split with the largest such derivative. This corresponds to maximizing 
$q(A, B)$. The complete algorithm using this rule is presented in Algorithm~\ref{alg:sfl-refined-max-q}.\footnote{Although preliminary tests suggested the maximum-$q$ heuristic performs well in practice, it was not adopted as the split-selection rule in the Bitcoin Core implementation.}

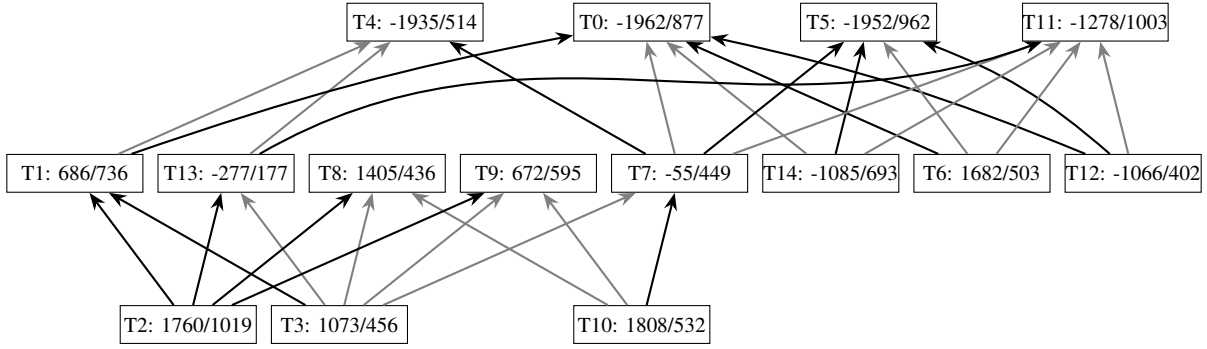
\begin{figure*} %
    \centering
    \begin{tikzpicture}[%
        every node/.style={draw, align=center, minimum width=1.8cm, minimum height=0.5cm, inner sep=0pt, font=\footnotesize},
        arrow/.style={-Stealth, thick},
        grayarrow/.style={-Stealth, thick, gray}
    ]

        \node (T4) at (4.5,4) {T4: -1935/514};
        \node (T0) at (7.5,4) {T0: -1962/877};
        \node (T5) at (10.5,4) {T5: -1952/962};
        \node (T11) at (13.5,4) {T11: -1278/1003};

        \node (T1) at (0,2) {T1: 686/736};
        \node (T13) at (2,2) {T13: -277/177};
        \node (T8) at (4,2) {T8: 1405/436};
        \node (T9) at (6,2) {T9: 672/595};
        \node (T7) at (8,2) {T7: -55/449};
        \node (T14) at (10,2) {T14: -1085/693};
        \node (T6) at (12,2) {T6: 1682/503};
        \node (T12) at (14,2) {T12: -1066/402};

        \node (T2) at (1.5,0) {T2: 1760/1019};
        \node (T3) at (3.5,0) {T3: 1073/456};
        \node (T10) at (7.5,0) {T10: 1808/532};
        
        \draw[arrow] (T12) to [bend right=10] (T5);
        \draw[arrow] (T12) to [bend right=5] (T0);
        \draw[grayarrow] (T12) to (T11);
        
        \draw[grayarrow] (T6) to (T5);
        \draw[grayarrow] (T6) to (T11);
        \draw[arrow] (T6) to (T0);

        \draw[arrow] (T14) to (T5);
        \draw[grayarrow] (T14) to (T11);
        \draw[grayarrow] (T14) to (T0);

        \draw[arrow] (T7) to (T4);
        \draw[grayarrow] (T7) to (T0);
        \draw[arrow] (T7) to (T5);
        \draw[grayarrow] (T7) to (T11);

        \draw[arrow] (T13) to [out=30, in=200] (T11);
        \draw[grayarrow] (T13) to (T4);
        
        \draw[arrow] (T1) to [bend left=4] (T0);
        \draw[grayarrow] (T1) to (T4);

        \draw[arrow] (T10) to (T7);
        \draw[grayarrow] (T10) to (T9);
        \draw[grayarrow] (T10) to (T8);

        \draw[arrow] (T3) to (T1);
        \draw[grayarrow] (T3) to (T13);
        \draw[grayarrow] (T3) to (T8);
        \draw[grayarrow] (T3) to (T9);
        \draw[grayarrow] (T3) to (T7);
        
        \draw[arrow] (T2) to (T1);
        \draw[arrow] (T2) to (T13);
        \draw[arrow] (T2) to (T8);
        \draw[arrow] (T2) to (T9);
        
    \end{tikzpicture}
    
    \caption{An example of a transaction DAG that can repeat the same sequence of split and merge steps when following a specific trajectory. A transaction with fee \( f \) and size \( s \) is denoted as \( f/s \).}
    \label{fig:sfl-termination-counterexample}
\end{figure*}

The refined SFL algorithm with the maximum-$q$ heuristic may revisit the same active forest state and therefore does not guarantee termination.
Fig.~\ref{fig:sfl-termination-counterexample} shows a cluster with 15 transactions and 30 dependencies that permits a cycle of 24 split-and-merge operations that ultimately returns to the original state. The black edges form the initial active tree:
\texttt{T1$\rightarrow$T0}, \texttt{T6$\rightarrow$T0}, \texttt{T12$\rightarrow$T0}, \texttt{T2$\rightarrow$T1}, 
\texttt{T3$\rightarrow$T1}, \texttt{T7$\rightarrow$T4}, \texttt{T7$\rightarrow$T5}, \texttt{T12$\rightarrow$T5}, 
\texttt{T14$\rightarrow$T5}, \texttt{T10$\rightarrow$T7}, \texttt{T2$\rightarrow$T8}, \texttt{T2$\rightarrow$T9}, 
\texttt{T13$\rightarrow$T11}, and \texttt{T2$\rightarrow$T13}.

From this configuration, a sequence of 24 steps can be performed that results in a return to the same state. Each step deactivates one edge (dependency) and activates another, sometimes a grey edge from the original graph. 
The steps are:

\begin{center}
\texttt{
-T12$\rightarrow$T0 +T13$\rightarrow$T4,
-T7$\rightarrow$T4 +T13$\rightarrow$T5,
-T2$\rightarrow$T1 +T3$\rightarrow$T8,
-T2$\rightarrow$T8 +T3$\rightarrow$T9, \\
-T2$\rightarrow$T9 +T6$\rightarrow$T5,
-T6$\rightarrow$T0 +T3$\rightarrow$T7, 
-T13$\rightarrow$T5 +T14$\rightarrow$T11,
-T14$\rightarrow$T5 +T12$\rightarrow$T11, \\
-T12$\rightarrow$T5 +T1$\rightarrow$T4,
-T13$\rightarrow$T4 +T3$\rightarrow$T13,
-T3$\rightarrow$T7 +T10$\rightarrow$T8,
-T3$\rightarrow$T8 +T10$\rightarrow$T9, \\
-T3$\rightarrow$T9 +T6$\rightarrow$T11,
-T6$\rightarrow$T5 +T7$\rightarrow$T11,
-T3$\rightarrow$T13 +T14$\rightarrow$T0,
-T14$\rightarrow$T11 +T12$\rightarrow$T0, \\
-T12$\rightarrow$T11 +T7$\rightarrow$T4,
-T1$\rightarrow$T4 +T7$\rightarrow$T0,
-T7$\rightarrow$T11 +T2$\rightarrow$T9,
-T10$\rightarrow$T9 +T2$\rightarrow$T8, \\
-T10$\rightarrow$T8 +T6$\rightarrow$T0,
-T6$\rightarrow$T11 +T2$\rightarrow$T1,
-T7$\rightarrow$T0 +T14$\rightarrow$T5,
-T14$\rightarrow$T0 +T12$\rightarrow$T5.
}
\end{center}

This demonstrates that Algorithm~\ref{alg:sfl-refined-max-q}, without additional safeguards, does not guarantee termination. 
While such cyclic states appear rare in practice, their existence motivates modifying the algorithm to guarantee termination.

\section{Linearization via Minimum Cuts}
\label{s:approach-ggt}

In 1982, Picard and Queyranne~\cite{Picard01111982} introduced an algorithm for finding a maximum-ratio closure in a graph. In our terminology, it consists of starting with an initial solution consisting of all transactions, and then repeatedly improving it by removing transactions from it, in such a way that the solution remains a closure, and the fee rate increases. With $n$ transactions, this will find the maximum fee rate closure in at most $n-1$ steps. We can use this approach to compute a complete optimal linearization, by finding
all consecutive maximum-ratio closures, of which there may be $n$. Together, this approach can find an optimal linearization using $\mathcal{O}(n^2)$ improvement steps.

Each improvement step in the paper's algorithm takes an existing closure as input, and tries to find in it a subset with higher fee rate that is also a closure. Specifically, if $\lambda$ is the fee rate of the existing closure, it defines $c_i = f_i - \lambda s_i$ for every transaction $i$ in it, i.e. its fee above the closure's average fee rate. It then finds the closure $S$ which maximizes $\sum_{i \in S} c_i$. If a subset closure exists with higher fee rate than $\lambda$, the solution will be one of them. If not, the solution will be the empty set or the entire input closure.

To achieve this, improvement steps work by finding a minimum cut of a flow network with $n+2$ vertices: one for each transaction, plus two special vertices $s$ (the source) and $t$ (the sink). The network is given the following edges: \begin{itemize}
\item For every dependency between two transactions, an edge with infinite capacity from the child to the parent.
\item For every transaction $i$, an edge with capacity $c_i$ from the source to the transaction if $c_i>0$, or an edge with capacity $-c_i$ from the transaction to the sink if $c_i<0$.
\end{itemize}
The minimum cut of this network will be a partition of the nodes into a source side $U$ (containing $s$) and a sink side $V$ (containing $t$), for which the sum of capacities of edges from $U$ to $V$ is minimal. The transaction nodes on the source side are our maximum fee rate closure.
It must clearly be a closure, because if it is not, the cut value would be infinite. It must further maximize $\sum_{i \in U} c_i$, because if a subset of transactions with positive total $c_i$ could be moved to $U$, or a subset with negative total $c_i$ could be moved to the $V$, while leaving $U$ a closure, it would decrease the cut value by that amount.

Many algorithms exist for solving this minimum-cut problem, and any of them can be used here. A very practical choice is the push-relabel algorithm~\cite{GT88Preflow} with maximum-label heuristic~\cite{cheriyan1988preflow}, which has complexity $\mathcal{O}(V^2 \sqrt{E})$, for $V$ vertices and $E$ edges. Instantiated for our problem this gives complexity $\mathcal{O}(n^4 \sqrt{m})$, for $n$ transactions and $m$ dependencies, to find the optimal linearization.

In 1989, Gallo, Grigoriadis, and Tarjan~\cite{gallo1989fast} introduced algorithms for finding minimum cuts in {\em parametrized} flow networks, allowing the computation of multiple minimum cuts with the same complexity as a single one, as long as all networks are the same, except some subgraphs may be contracted into the source and/or sink, and the capacities may vary monotonically in a single parameter. This is exactly the case in our setting, where earlier chunks are contracted into the source, and the parameter $\lambda$ is the fee rate of the best known next closure so far. This is achieved by reusing the push-relabel data structure across multiple instances of the minimum-cut problem, in such a way that the complexity for up to $n$ such instances is the same as that for a single instance. This reduces the complexity for a full linearization to $\mathcal{O}(n^3 \sqrt{m})$.

The all-breakpoints algorithm in the same paper is a further improvement over this. Instead of computing the maximum fee rate closures one by one in order, it uses a divide-and-conquer approach, where every computed minimum cut splits the graph in two, and both sides are kept and then divided further with more cuts. It runs every minimum cut twice, in parallel, one forward and one in reverse, and stops both when the first one finishes. A clever scheme then reuses the final state for one of the two resulting subgraphs, in such a way that the complexity of the overall algorithm is the same as that of a single minimum cut, $\mathcal{O}(n^2 \sqrt{m})$. As the number of dependencies is bounded by $n^2/4$, this means $\mathcal{O}(n^3)$ when expressed just in $n$. This is referred to in our benchmarks as the bidirectional variant. A simpler variant, where a random direction is chosen for every minimum cut, has worst-case complexity $\mathcal{O}(n^3 \sqrt{n})$, but has better constant factors, and we believe it is also $\mathcal{O}(n^2 \sqrt{n})$ when averaged over random seeds.

Note that the paper gives a somewhat better complexity $\mathcal{O}(nm \log(n^2/m))$ for a minimum cut, by relying on a dynamic trees-based algorithm for minimum cuts instead of the max-label heuristic. In practice that algorithm is complicated and has bad constant factors, so it is rarely used. However, an instantiation with max-label is very practical, and an obvious choice for comparison with SFL, as it has great worst-case complexity. 

\end{document}